\documentclass[a4paper,leqno,10pt]{article}

\usepackage{ams math,amssymb,amsfonts,amsthm,mathrsfs,MnSymbol}
\usepackage{graphicx,epsfig,color}

\usepackage{setspace}
\usepackage{latexsym}
\usepackage{epsfig}
\usepackage{yfonts}

\usepackage{enumitem}
\setitemize{wide}

\newtheorem{Definition}{Definition}
\newtheorem{Proposition}{Proposition}
\newtheorem{Theorem}{Theorem}

\newtheorem{Lemma}{Lemma}

\newtheorem{Corollary}{Corollary}

\newtheorem{Aux-Lemma}{Aux-Lemma}

\newcommand{\vs}{\vspace{0.2cm}}
\newcommand{\n}{\noindent}
\newcommand{\be}{\begin{equation}}
\newcommand{\ee}{\end{equation}}
\newcommand{\ben}{\begin{equation*}}
\newcommand{\een}{\end{equation*}}

\newcommand{\dist}{{\rm{dist}}}

\usepackage{etoolbox}

\usepackage{fancyhdr}
\AtBeginDocument{\thispagestyle{plain}}
\fancypagestyle{plain}{
\fancyhead{}
\fancyfoot{}
\fancyhead[LE,RO]{}
\fancyhead[LO,RE]{}
\fancyfoot[R]{\thepage}

}

\newcommand{\br}{\bar{r}}
\newcommand{\AR}{{\mathcal A}_{\mathbb{R}^{3}}}
\newcommand{\length}{{\rm{length}}}
\newcommand{\const}{\varrho}

\begin{document}

\thispagestyle{plain}

%\n {\fontsize{18pt}{13pt}\selectfont \sc Stationary solutions and Asymptotic flatness I.}

\begin{center}
{\sc\LARGE Stationary solutions and

\vspace{.3cm}
Asymptotic flatness I}

\vspace{.3cm}
\n {Martin Reiris}\\
\vspace{.2cm}
{\small email: martin@aei.mpg.de}\\

\vspace{.1cm}
\n \textsc{Max Planck Institute f\"ur Gravitationsphysik \\ Golm - Germany}\\

\vspace{0.6cm}
\n \begin{minipage}[l]{11cm}
\begin{spacing}{1}
{\small In this article and its sequel we discuss the asymptotic structure of space-times representing isolated bodies in General Relativity. 
Such space-times are usually required to be asymptotically flat (AF), and thus to have a prescribed type of asymptotic.
Despite all the ``reasonable" that the requirement is, it seems to be against the spirit of General Relativity where the global structure of the space-time should be also considered as a variable. It is shown here that, even eliminating from the definition any a priori reference or assumption about the asymptotic, the space-times of isolated bodies are unavoidably and a posteriori AF. 
In precise terms, between the two articles it is proved that any vacuum strictly stationary space-time end whose (quotient) manifold is diffeomorphic to $\mathbb{R}^{3}$ minus a ball and whose Killing field has its norm bounded away from zero
is necessarily AF with Schwarzschidian fall off. The ``excised" ball would contain (if any) the actual material body, but this information or any other is not necessary to reach the conclusion. Physical and mathematical implications are also discussed. 

\vs

}
\end{spacing}
\n \end{minipage}

%\vspace{.6cm}
\end{center}
\section{Introduction.}

In this section we discuss with certain freedom the physical motivations of this article.
Around 13 billions of years ago the first galaxies started to form. Since then they continued accreting matter, delimited their visible shapes and, as the volume of the universe expanded, they drifted apart.    
Then the profiles of their gravitational potentials settled and became a distinguishable imprint of their material contents. One can imagine that along this journey the space surrounding a given galaxy decomposes naturally into a bulk, a far field zone and, farther away, the outside world \cite{MR2441850}. 
Moreover in this landscape the outside world is so far away that to model the far field zone we could simply prescind of it and replace it by an essentially flat empty space. 
In more pragmatic words, at the end of this eternal expansion the far field zone can be modeled as an empty and AF stationary region of space. This (a bit romantic) story is based in facts and also in intuitions.
But, does it have to be like that? More concretely, could the far field zone of isolated galaxies reach eventually a different type of asymptotic? The quest, which at first sight seems to be only of an academic motivation, has indeed some physical relevance. 
The reason is that galaxies in their current stage do not accompany entirely this picture. 

Let us bring one among the many astronomical data available into the discussion. Observations show indisputably that far outside their visible regions the rotational velocity of stars and satellite galaxies around the center of disc galaxies remains remarkably constant in the radius (see \cite{Sofue:2000jx} and, more recently, \cite{Persic:1995ru}, \cite{Salucci:2007tm}). For a typical disc galaxy like the Milky Way, the rotational velocities are of the order of 250 km/s. The problem about this behavior is that it is against a flat asymptotic (for a classical discussion of isolated systems see \cite{Ehlers1980}). Of course the usual (and plausible) explanation claims the existence of huge haloes composed of weakly interacting dark matter particles enclosing the galaxies and causing such gravitational distortions. It is not the purpose of this article to adhere to or to refute this belief (we do not have the background to do so) but rather to investigate what General Relativity says about the asymptotic of isolated systems. Despite of this we will adventure a physical incursion at the end of the introduction. 

A simple static perfect fluid solution of the Einstein equations displaying such flat rotation curves around the origin can be given explicitly as follows. The metric (in geometrized units) is given by
\ben
{\bf g}=-r^{\displaystyle 2v^{2}}dt^{2}+ (1+4\pi v^{2}(2-v^{2}))dr^{2}+r^{2}d\Omega^{2}
\een
where, as a simple calculation shows, $v$ is the rotational velocity of circular orbits around the origin $r=0$ (hence constant). 
The stress-energy tensor is that of a perfect fluid with $\rho=\rho_{0}/r^{2}$ and $p=p_{0}/r^{2}$ where 
\ben
\rho_{0}=\frac{1}{8\pi+{\displaystyle \frac{2}{(2-v^{2})v^{2}}}}\qquad \text{and}\qquad p_{0}=\frac{v^{4}}{2+8\pi(2-v^{2})v^{2}}
\een
This space-time is found by solving the Tolman-Oppenheimer-Volkoff equation of hydrostatic equilibrium \cite{MR757180} with the ansatz $\rho=\rho_{0}/r^{2}$ and $p=p_{0}/r^{2}$ and has a number of interesting properties. It is spherically symmetric, the Killing field $\partial_{r}$ is static and $r\partial_{r}+(1-v^{2})t\partial_{t}$ is a conformal Killing field. Moreover $r\partial_{r}$ is a homothetic Killing field of the slice $t=0$ which is hence self-similar. When $v\sim 0$ then $v\sim \sqrt{m(r)/r}$ where is $m(r)$ is the mass contained until the coordinate radius $r$ (on the $t=0$ slice). Moreover the energy density
$\rho_{0}$ measures the area defect as we have $8\pi \rho_{0}=(A_{E}(\bar{r})-A(\bar{r}))/A_{E}(\bar{r})$ where here $A_{E}(\bar{r})=4\pi \bar{r}^{2}$ is the Euclidean area of a sphere of radius $\bar{r}$ and $A(\bar{r})$ is the area of the sphere of metric radius $\bar{r}$ on the $t=0$ slice (i.e. $\bar{r}=\bar{r}(r)$ is the physical distance from the sphere of coordinate radius $r$ to the origin $r=0$). When $v\sim 0$ then (in natural units) $(v/c)^{2}\sim (A_{E}(\bar{r})-A(\bar{r}))/2A_{E}(\bar{r})$. In particular, for rotational velocities of the order of $300$ km/s the area defect is of the order of $2\, 10^{-6}$. In any case the area defect does not depend on the radius is a measure of the non-asymptotic flatness. Finally it is worth remaking that in the regime of small rotational velocities this model reduces to the so called {\it singular isothermal model} which is widely used in Galactic dynamics \cite{MR7571801}. 

\begin{figure}[h]
\centering
\includegraphics[width=10cm,height=5cm]{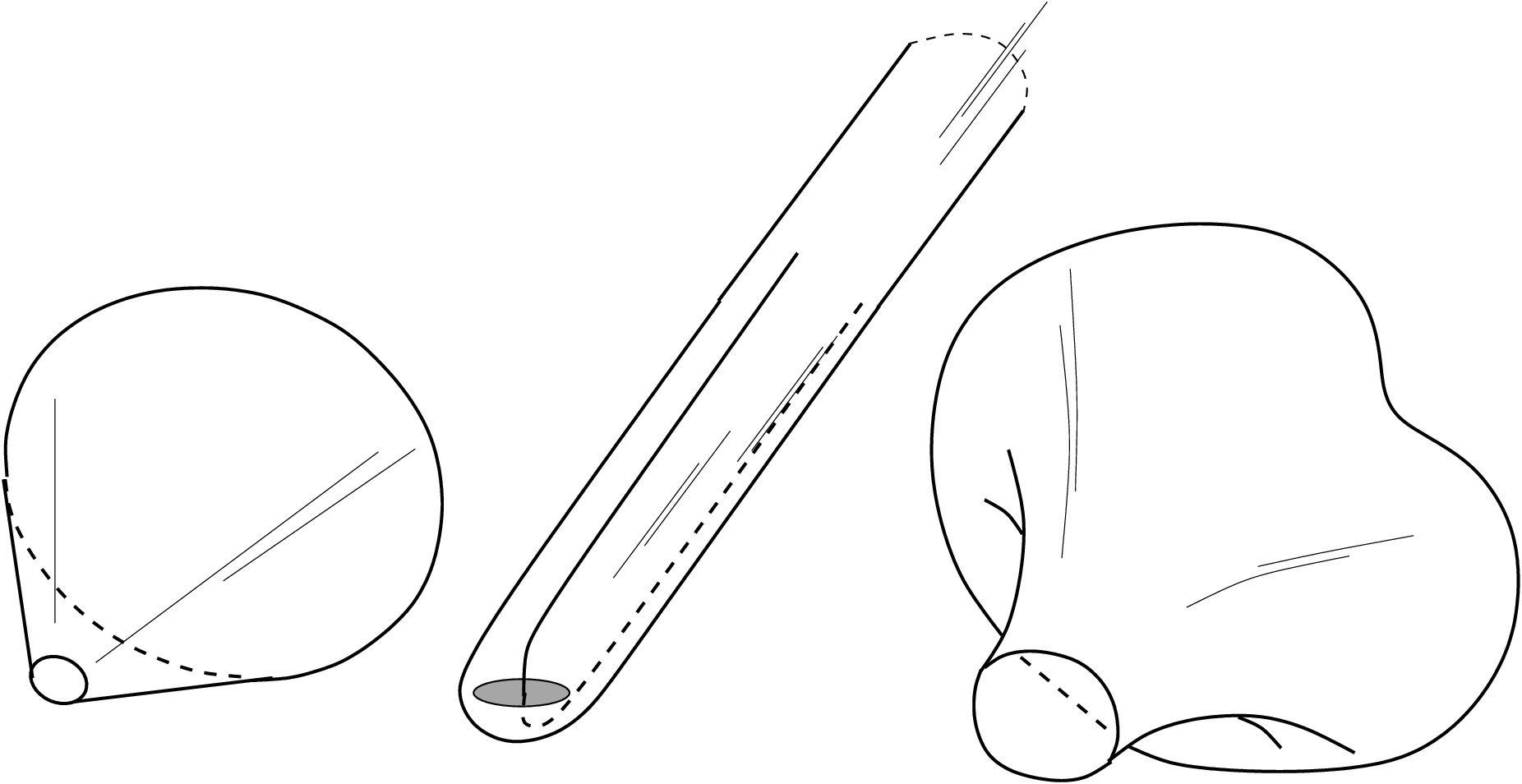}
\caption{Representation of the space-like sections of the three different asymptotic discussed in the introduction. On the left is a self similar three-geometry. Centered is represented a rotating disc having the cylindrical asymptotic of a near-horizon geometry. On the right is an asymptotically flat three-geometry.}
\label{Fig1}
\end{figure} 

The previous is an interesting and relatively appealing stationary space-time filled with matter but not AF. But what about vacuum space-times representing the stationary exterior of isolated systems? The paradigmatic case is of course the Schwarzschild solution which is AF. The Kerr black-holes can also be considered to represent isolated bodies (black-holes) and there are other exact solutions, like the Tomimatsu-Sato class, that could also account for the exterior of isolated bodies. Numerical solutions, also AF, have been studied extensively in the literature \cite{MR2441850}. But, are there systems generating a non-asymptotically flat stationary space-time? Interestingly the answer is yes. As R. Mainel and G. Neugebauer have shown (see for instance \cite{MR2441850}), when a disc of dust rotates at a particular rate it generates a stationary space-time that is not AF but rather asymptotically cylindrical. More precisely far away from the disc the space-time approaches the so called near horizon solution which has an appropriate space-like section displaying a cylindrical geometry (see Figure \ref{Fig1}). 
Concretely such slice of the near horizon geometry is diffeomorphic to $S^{2}\times \mathbb{R}$ and the  initial data on it (normalized to have Komar angular momentum equal to one) has metric $g$ and the second fundamental form $K$ given by
\ben
\left\{\begin{array}{l}
g= \frac{\displaystyle 4\sin^{2}\theta}{\displaystyle 1+\cos^{2}\theta}d\varphi^{2}+(1+\cos^{2}\theta)d\theta^{2} +(1+\cos^{2}\theta) d r^{2},\vs\vs\\
K=\frac{\displaystyle 2\sin^{2}\theta}{\displaystyle (1+\cos^{2}\theta)^{\frac{3}{2}}} (d\varphi d r+d r d\varphi)
\end{array}
\right.
\een
where $(\theta,\varphi)$ are the coordinates in $S^{2}$ and $r$ is the coordinate in the $\mathbb{R}$-factor. The metric has a translational symmetry along $r$ and is therefore cylindrical. The space-time metric of its globally hyperbolic development is
\begin{align*}
{\bf g}=&-\big(1+\cos^{2}\theta\big)d t^{2}+\big(1+\cos^{2}\theta\big)\big(r \tan t\, d t +d r\big)^{2}
\\
\nonumber &+\frac{4\sin^{2} \theta}{(1+\cos^{2}\theta)}\big(r d t-\varphi\big)^{2}+\big(1+\cos^{2}\theta\big)^{2}d \theta^{2}
\end{align*}
which contains a Cauchy horizon at $t=\pi/2$ and is not geodesically complete. This space-time has the peculiar property that its ergoregion (i.e. the region where the stationary Killing field is space-like) extends to infinity. However, it is not this but geodesically incompleteness what is a serious drawback to qualify as the far field metric of an ``isolated body" solution. Should one include in the definition of isolated body space-time also geodesically completeness? In many respects this is a natural condition. 

In these articles we are able to describe the asymptotic of stationary space-times representing isolated systems and having two characteristics. On one side we require the stationary Killing field to be time-like and to be bounded in norm away from zero outside an arbitrarily large but finite region surrounding the object. On the other hand we require space-time geodesic completeness (until the boundary) also in the exterior region. 
In this scenario we show that the space-time is necessarily AF with Schwarzschidian fall off. The method of proof of this result is such, that we can prescind altogether of the details of the material body, even of its existence, and work completely in the exterior vacuum region. These ``exterior" vacuum space-times are defined in Definition \ref{DEFEND} and called Strongly Stationary. 
All this is put in Section \ref{THEPRESTA} into a formal mathematical setup. 

To finish the introduction we would like to adventure a physical sentence. It appears from our findings that, granting General Relativity to be correct and granting the stationarity of the exterior regions of galaxies, then the existence of dark matter seems to be an inevitable fact. In other words it is not possible to explain the observed gravitational distortions around disc galaxies in terms of vacuum General Relativity alone. 

\subsection{The setup.}\label{THEPRESTA}

Strongly stationary ends are defined as follows.

\begin{Definition}\label{DEFEND} Let $({\bf M},{\bf g})$ be a smooth chronological solution of the vacuum Einstein equations with smooth boundary. Then, $({\bf M},{\bf g})$ is said to be a stationary space-time end if the following conditions are fulfilled, 
\begin{enumerate}[labelindent=\parindent,leftmargin=*,label={\rm (Q\arabic*)}]
\item[\rm (S1)] There is a time-like complete Killing field $X$ in ${\bf M}$ tangent to $\partial {\bf M}$ at $\partial {\bf M}$ such that the quotient of ${\bf M}$ by the orbits of $X$ is diffeomorphic to $\mathbb{R}^{3}$ minus an open ball,
\item[\rm (S2)] $({\bf M},{\bf g})$ is geodesically complete until the boundary, namely, geodesics either end at $\partial {\bf M}$ or are defined for infinite parametric time.
\end{enumerate}
A stationary space-time end $({\bf M},{\bf g})$ is said to be strongly stationary if in addition
\begin{enumerate}[labelindent=\parindent,leftmargin=*,label={\rm (Q\arabic*)}]
\item[\rm (S3)] There is a positive constant $c$ such that $-\langle X,X\rangle_{\bf g}\geq c$ all over ${\bf M}$.
\end{enumerate}
\end{Definition}

\n Above $\langle\ ,\ \rangle_{\bf g}$ is the ${\bf g}$-inner product. In the following discussion we summarize the mathematics of strong stationary ends $({\bf M},{\bf g})$ as seen in the quotient space. 
We refer to \cite{MR2003646} for a detailed account on stationary solutions. 

Let $E$ be the manifold that result from the quotient of ${\bf M}$ by $X$, let $\mathcal{\pi}:{\bf M}\rightarrow E$ be the projection and let $\tilde{g}$ be the quotient three-metric. By (S1) $E$ is diffeomorphic to $\mathbb{R}^{3}$ minus an open ball and by (S2) $(E,\tilde{g})$ is geodesically complete until the boundary (recall that geodesics in $(E,\tilde{g})$ can be lifted to geodesics in $({\bf M},{\bf g})$ perpendicular to $X$ and preserving the arc-length).
All the relevant components of the Einstein equations can be written in the quotient space in terms of $\tilde{g}$, 
the ``norm" $u$ of the Killing $X$, i.e. $u:=\sqrt{-\langle X,X\rangle_{\bf g}}$, and the twist one-form $\omega$ which is defined by 
\ben
\omega=\pi_{*}\big(\scalebox{.6}{$\bigstar$}(\frac{1}{2}\xi\wedge {\rm d}\, \xi\, )\, \big)
\een
where $\scalebox{.6}{$\bigstar$}$ is the ${\bf g}$-Hodge star and the one-form $\xi$ is the ${\bf g}$-dual of the Killing, i.e. $\xi:=\langle X,\  \rangle_{\bf g}$. 
In ${\bf M}$ the metric ${\bf g}$ is ${\bf g}=-(\xi\otimes \xi)/u^{2}+\pi^{*} g$ and in terms of the data $(g,\omega,u)$ an isommetric copy $(\hat{\bf M}, {\hat{\bf g}})$ of $({\bf M},{\bf g})$ can be obtained by making $\hat{\bf M}=\mathbb{R}\times E$ and $\hat{\bf g}=-u^{2}({\rm d}t+\theta)^{2}+\tilde{g}$, where here $t$ is the coordinate in the $\mathbb{R}$-factor and the one form $\theta$ (in $E$) is found by solving ${\rm d}\theta = - \star 2\omega/u^{4}$ where $\star$ is the $\tilde{g}$-Hodge star. 

In principle one can work with the variables $(\tilde{g},\omega,u)$ but, as it turns out, the Einstein equations display a rich structure when expressed instead in terms of the conformally transformed metric 
\ben
g:=u^{2}\tilde{g},
\een
the form $\omega$ and the function $u$. In terms of $(g,\omega,u)$ the vacuum Einstein equations are equivalent to (\cite{MR2003646})
\be\label{MEE}
\left\{
\begin{array}{l}
\ Ric=2\, \nabla \ln u \otimes \nabla \ln u +\frac{\displaystyle 2}{\displaystyle u^{4}}\, {\displaystyle \omega\otimes \omega},\\
\ \Delta \ln u=-\frac{\displaystyle 2\, \vert\, \omega\, \rvert^{2}}{\displaystyle u^{4}},\vs\\
\ {\rm div}\, \omega=4\langle \nabla \ln u, \omega \rangle,\vs\\
\ {\rm d}\, \omega=0
\end{array}
\right.
\ee
where $Ric$ is the Ricci tensor of $g$, $\Delta$ is the $g$-Laplacian, $|\ \ |$ denotes $g$-norm and $\langle\ ,\ \rangle$ denotes the $g$-inner product. In the last two equations $div\, \omega$ is the $g$-divergence of $\omega$ and $d\, \omega$ its exterior derivative.
The equations (\ref{MEE}) enjoy remarkable structures which will be introduced however as the article progresses.    
The data $(E; \tilde{g},\omega,u)$ or the equivalent data $(E;g,\omega,u)$ will be called a {\it stationary end} when $({\bf M},{\bf g})$ is a stationary end and a {\it strong stationary end} when $({\bf M},{\bf g})$ is strongly stationary.

By (S3) the space $(E,g)$ of a strongly stationary end is also geodesically complete (until the boundary). As a matter of fact, 
the assumption (S3) is introduced to guarantee the completeness of $(E,g)$. The analysis and the results of this article remain unchanged if instead of (S3) we impose just the completeness of $(E,g)$. We will recall this later when we comment Corollary \ref{COROLLL}.

To be explicit, the definition of asymptotically flatness with Schwarzschidian fall off that we adopt is the following (c.f. \cite{MR1314057}) 

\begin{Definition}\label{AFDEF} Let $(E; \tilde{g},\omega,u)$ be a stationary end. Then, it is asymptotically flat with Schwazschidian fall off if there is a coordinate system $\{x=(x^{1},x^{2},x^{3})\}$ covering $E$ up to a compact set such that
\ben 
\big|\, \delta_{ij}-\tilde{g}_{ij}\, \big|\leq K/|x|,\qquad \big|\, \partial_{k} \tilde{g}_{ij}\, \big|\leq K/|x|^{2},\qquad \big|\partial_{m}\partial_{k}\, \tilde{g}_{ij}\, \big|\leq K/|x|^{3},
\een
and,
\ben
\big|\, \partial_{i} u\, \big|+\big|\, \omega_{i}\, \big| \leq K/|x|^{2}
\een
plus further progressive power-law decay for the norms of the multiple $\partial$-derivatives of $\tilde{g}_{ij}$, $u$ and $\omega_{i}$, where here $K$ is a positive constant and $|x|$ is the norm of $x=(x^{1},x^{2},x^{3})$ as a vector in $\mathbb{R}^{3}$.
\end{Definition}

The definition is the same if instead of $(E;\tilde{g},\omega,u)$ we consider $(E;g,\omega,u)$. Before stating the main Theorem let us recall the definition of cubic volume growth. Denote by ${\mathcal T}_{g}(\partial E,r)$ the metric (tubular) neighborhood of $\partial E$ and radius $r>0$, that is, the set of points in $E$ at a distance (see Section \ref{NOTATION}) from $\partial E$  less than $r$. Then, $(E,g)$ is said to have cubic if $\lim Vol_{g}\big({\mathcal T}_{g}(\partial E,r)\big)/r^{3}=\mu>0$. Note that the limit always exists as a result of the Bishop-Gromov monotonicity of $Vol_{g}\big({\mathcal T}_{g}(\partial E,r)\big)/r^{3}$ due to the non-negativity of the Ricci curvature of $g$ (first equation in (\ref{MEE})). 

The purpose of this article and its sequel is then to prove:
\begin{Theorem}\label{MAINT}
Any strongly stationary end $(E; g,\omega,u)$ having cubic volume growth is asymptotically flat with Schwarzschidian fall off.
\end{Theorem}
Although we will keep including explicitly inside the statements that the strongly stationary solutions have cubic volume growth, it is a very important fact that this condition can be entirely removed due to very general geometric facts [arXiv:1212.1317]. In this point we found it important to distinguish between what requires the full structure of stationary solutions and what is indeed a property of a much more general character which as a matter of fact requires for its proof quite different techniques. 
The reason why this condition is unnecessary is the following. First, as remarked before, strongly stationary ends have non-negative Ricci curvature , and as proved by M. T. Anderson \cite{MR1806984}, they have also quadratic curvature decay. 
Second it was proved in [arXiv:1212.1317] that complete metrics in $\mathbb{R}^{3}$ minus a ball with non-negative Ricci curvature and quadratic curvature decay have cubic volume growth. 
The combination of these two facts and Theorem \ref{MAINT} shows that assuming cubic volume growth is unnecessary. We state this as a corollary to the Theorem \ref{MAINT}, this time in terms of the physical variables $(\tilde{g},\omega,u)$. The corollary is the most important consequence emerging out of this research.
\begin{Corollary}\label{MAINTI}
Any strongly stationary end $(E;\tilde{g},\omega,u)$ is asymptotically flat and has Schwarzschidian fall off.
\end{Corollary}

A relevant open question is whether stationary space-time ends are always strongly stationary or not. Or, in the light of Corollary \ref{MAINTI}, it is open whether stationary space-time ends are always AF or not. The following ``Gap Corollary" partially answers this question. To size the relevance of the result keep in mind the a priori estimates from \cite{MR1806984} according to which   
there is a universal constant ${\mathcal K}>0$ (so far unknown) such that for any strongly stationary end $(E;\tilde{g},\omega,u)$ we have 
\ben
\big|(\nabla \ln u)(p)\big|_{\tilde{g}}\leq \frac{\mathcal K}{\dist_{\tilde{g}}(p,\partial E)},\quad \big|\, \omega(p)\, \big|_{\tilde{g}}\leq \frac{\mathcal K}{\dist_{\tilde{g}}(p,\partial E)}\quad \text{and}\quad \big|Ric_{\tilde{g}}(p)\big|_{\tilde{g}}\leq \frac{\mathcal K}{\dist^{2}_{\tilde{g}}(p,\partial E)}
\een 
where $\dist_{\tilde{g}}(p,\partial E)$ is the $\tilde{g}$-distance from $p$ to $\partial E$ (see Section \ref{NOTATION}). 
\begin{Corollary}\label{COROLLL} 
%There is a universal constant ${\mathcal K}>0$ such that for any strongly stationary end $(E;\tilde{g},\omega,u)$ we have $|Ric_{\tilde{g}}(p)|\leq {\mathcal K}/d_{\tilde{g}}(p)^{2}$. 
Let $(E;\tilde{g},\omega,u)$ be a stationary end. If 
\be\label{DECAYU}
\big|(\nabla \ln u) (p)\big|_{\tilde{g}}\leq \frac{1}{\dist_{\tilde{g}}(p,\partial E)}
\ee
then the end is asymptotically flat with Schwazschidian fall off. 
\end{Corollary} 
\n This shows that if one can prove that the universal constant ${\mathcal K}$ must be a priori less or equal than one then every stationary end is AF with Schwarzschidian fall off. The question is undoubtedly of fundamental importance.
The proof of the Corollary, whose details are left to the reader, is done by showing that under (\ref{DECAYU}) we have $u(p)\geq c/\dist_{\tilde{g}}(p,\partial E)$ for some constant $c>0$ and for all $p$ with $\dist_{\tilde{g}}(p,\partial E)\geq 1$, and then observing from this that if $\tilde{g}$ is complete then so is $g=u^{2}\tilde{g}$. As commented before this is enough to get the same conclusions as in Corollary \ref{MAINT}. 

Before passing to the more technical sections let us glance on the contents of each of the two articles. 

In this first article (Part I) we work with Weakly Asymptotically Flat (WAF) stationary ends and prove that WAF ends are AF with Schwarzschidian decay. The notion of WAF end is a generalization of that of ``weakly decaying solutions" as defined by D. Kennefick and N. \'O Murchadha in \cite{MR1314057}. Roughly speaking we are defining as much as could be allowed a notion of ``decaying into flat" without relying on any coordinate system and power-law decay for the curvature. Incidentally, by proving that WAF ends are AF we are answering a question raised in  {\it Remark 2} of \cite{MR1314057} which (quoted) says: ``The $\delta < 0$ condition classically means that we are willing to consider any metric which decays to flat space as $1/r^{\varepsilon}$ for any $\varepsilon$ . Obviously one would like to replace this with `going flat' and not require any kind of power law decay. It is difficult to see how this might be achieved; none of the present battery of weighted spaces (classical, Holder, .., ) seem suitable." [\footnote{We thanks H. Friedrich for pointing out this reference.}]. 
The definition of WAF ends is given in Section \ref{WAFECHP}. In Section \ref{PREWAFE} we work out the main properties of WAF ends and reach the conclusion that to prove that they are AF it is indeed enough to prove that the so called $\varepsilon$-flat ends, which enjoy much nicer properties, are AF. Finally in Section \ref{SKTD} it is proved the main result of the Part I, namely
\begin{Theorem}\label{FINALT}
Every weakly asymptotically flat end is asymptotically flat with Schwarzschidian fall off.
\end{Theorem}

In the second article (Part II) instead we work with strongly stationary ends and prove
\begin{Theorem}
Every strongly stationary end having cubic volume growth is weakly asymptotically flat. 
\end{Theorem}
Hence, after the two articles we would have proved, as explained before, that strongly stationary ends are AF with Schwarzschidian fall off.

\vs
\section{Background material I}\label{NOTATION}

We collect here the material required for the technical discussions. We introduce too the most relevant terminology and notation. From now on our main variables will be $(g,\omega,u)$.

\vs
{\sc Distance.} 

\vs
- The distance between two points $p$ and $q$ in a connected manifold $(M,g)$ is $\dist_{g}(p,q)=\inf\big\{\length_{g}(\mathscr{C}_{p,q}),\ \mathscr{C}_{p,q}\ \text{a}\ C^{1}\ \text{curve in}\ M\ \text{joining}\ p\ \text{to}\ q\big\}$. $(M,g)$ is said complete if $(M,\dist_{g})$ is complete as a metric space. 
The distance from a point $p$ to a set $\Omega\subset M$ will be denoted by $\dist_{g}(p,\Omega)=\big\{\dist_{g}(p,q),q\in \Omega\big\}$. More general the distance between two sets $\Omega_{1}$ and $\Omega_{2}$ is denoted by $\dist_{g}(\Omega_{1},\Omega_{2})=\inf\big\{\dist_{g}(p,q),p\in\Omega_{1}, q\in \Omega_{2}\big\}$ [\footnote{Properly speaking this is not a metric in the subsets of $M$. In particular the distance is zero if for instance they share a point but are different sets.}].

- The metric induced on stationary ends $(E,g)$ will be noted by $\dist(p,q)$ and always without the subindex $g$. The distance function to the boundary $\partial E$ of stationary ends will be denoted with total exclusivity by $d(p)$ or simply $d$, that is, $d(p)=\dist(p,\partial E)=\inf\big\{\dist(p,q),q\in \partial E\big\}$.  

\vs
{\sc Scaling.} 

\vs
- For any $r>0$ we will denote by $g_{r}$ to the scaled metric
\ben
g_{r}:=\frac{1}{r^{2}}\, g
\een
Tensors and metric quantities constructed out of $g_{r}$ will be sub-indexed with an $r$. For instance, for the scalar curvature we have $R_{g_{r}}=R_{r}=R/r^{2}$ and for the Ricci curvature $Ric_{g_{r}}=Ric_{r}=Ric$ (although $Ric_{r}=Ric$ we will keep including the subindex $r$). Also, $d_{r}(p)=d(p)/r$. {\it This way of notating will be used extensively all through the article and is crucial keeping track of it.}

\vs
{\sc Area, second fundamental form and mean curvature.} 

\vs
- The Riemannian metric induced on compact embedded surfaces $S\subset E$ will be denoted by $h$ and the $h$-area of $S$ by $A(S)$. Following the notation introduced before, the metric induced in $S$ from $g_{r}$ is denoted by $h_{r}$ and the $h_{r}$-area of $S$, i.e. $A(S)/r^{2}$, is denoted by $A_{r}(S)$.   

- The second fundamental form of $S$ (fixed some normal) will be denoted by $\Theta$ and the mean curvature ${\rm tr}_{h}\Theta$ by $\theta$. Again, for the second fundamental form and the mean curvature of $S$ found from the scaled metric $g_{r}$ we will use $\Theta_{r}(=r\Theta)$ and $\theta_{r}(=\theta/r)$ respectively.  

\vs
{\sc Annuli and metric annuli.} 

\vs
-  For any $0<a<b$ we will denote by ${\mathcal A}(a,b)$ (resp. ${\mathcal A}[a,b]$) the set
\ben
{\mathcal A}(a,b)=\big\{p\in E/a<d(p)<b\big\},\quad \text{(resp.}\ {\mathcal A}[a,b]=\big\{p\in E/a\leq d(p)\leq b\big\}\text{)} 
\een
and call it {\it the open (resp. closed) metric annulus of radii $a$ and $b$}. The notation ${\mathcal A}(a,b)$ (resp. ${\mathcal A}[a,b]$) will always refer to open (resp. closed) metric annuli defined with respect to the unscaled metric $g$ but the subindex $r$ is included when the (open or closed) metric annuli are defined with respect to the scaled metric $g_{r}=g/r^{2}$, namely
\ben
{\mathcal A}_{r}(a,b)=\big\{p\in E/ a<d_{r}(p)<b\big\}\quad \text{and}\quad {\mathcal A}_{r}[a,b]=\big\{p\in E/ a\leq d_{r}(p)\leq b\big\}
\een
This is consistent with the notation introduced before. Note that for all $r>0$ we have ${\mathcal A}(ar,br)={\mathcal A}_{r}(a,b)$ and ${\mathcal A}[ar,br]={\mathcal A}_{r}[a,b]$. 

- Standard open annuli in $\mathbb{R}^{3}$ will be denoted by $\AR(a,b)$, namely,
\ben
\AR(a,b)=\big\{x\in\mathbb{R}^{3},a<|x|<b\big\} = B_{\mathbb{R}^{3}}(o,b)\setminus \overline{B_{\mathbb{R}^{3}}(o,a)}
\een
where for any $c>0$, $B_{\mathbb{R}^{3}}(o,c)$ is the open ball of center the origin $o=(0,0,0)$ and radius $c$ in $\mathbb{R}^{3}$.
As before closed annulus in $\mathbb{R}^{3}$ are denoted by $\AR[a,b]=\big\{x\in \mathbb{R}^{3},a\leq |x|\leq b\big\}$. 

-A manifold $\Omega$ is an open (resp. closed) annulus if $\Omega$ is diffeomorphic to $\AR(1,2)$ (resp. $\AR[1,2]$). 
A metric annulus doesn't have to be necessarily an open annulus in this sense. In general, the shape of the metric annuli can be wild. 

All these notations will be used extensively.

\vs
{\sc The Ernst equation.} 

\vs
- The second and third equations of (\ref{MEE}) can be grouped in the so called Ernst equation. Let $\phi$ be a potential for $\omega$, i.e. $d\phi=\omega$. Define the complex function ${\mathcal E}=u^{2}+2\phi i$. Then the Ernst equation is
\be\label{ERNST0}
\Delta {\mathcal E}=2\frac{\displaystyle \langle\nabla {\mathcal E},\nabla {\mathcal E}\rangle}{\displaystyle {\mathcal E}+{\mathcal E}^{*}}
\ee
where ${\mathcal E}^{*}$ is the complex conjugate. Making $\chi:=\nabla {\mathcal E}$ and $\zeta:= (\nabla {\mathcal E})/u^{2}$, the equation (\ref{ERNST0}) can be written in the form of the linear system in $\chi$,
\be\label{LSCH}
\left\{
\begin{array}{l}
div\, \chi = \langle\zeta,\chi\rangle,\\
d\, \chi=0
\end{array}
\right.
\ee
where we are thinking $\zeta$ simply as a coefficient. This viewpoint of the Ernst equation will be important and will be recalled in the proof of the Proposition (\ref{PREL}).

\vs
{\sc Curvature.} 

\vs
- A first fundamental property derived from the first equation in (\ref{MEE}) is
\ben
|\, Ric\, |^{\, 2}\leq R^{\, 2}\leq 2|\, Ric\, |^{\, 2}
\een
which says that the scalar curvature $R$ fully controls the Ricci curvature $Ric$. Here $|Ric|^{2}=Ric_{ab}Ric^{ab}$. 
A second fundamental property is
\be\label{SUPERHAR}
\Delta R\geq R^{2}
\ee
which says that the scalar curvature is subharmonic. The proof of (\ref{SUPERHAR}) is given in \cite{MR1806984} (pg. 987-988) by manipulating the Bochner-type of formula for the energy density $e=R$ of the harmonic map ${\mathcal E}=(2\phi,u^{2})$ from $E$ into the half space model of the hyperbolic two-space $\mathbb{H}^{2}=\big(\{(x,y),y>0\},({\rm d}x^{2}+{\rm d}y^{2})/y^{2}\big)$.

- Another essential property of the curvature of stationary solutions is M. T. Anderson's a priori curvature decay \cite{MR1806984}. Applied to ends it says that there is a universal constant ${\mathcal K}>0$ such that for any strong stationary end $E$ we have [\footnote{There is a caveat here. The curvature estimate provided in {\bf Theorem 0.2} of \cite{MR1806984} is (as written) for the space-time metric and not for the metric $g$. However the proof of that Theorem is achieved by proving first the estimate $|Ric_{g}(p)|\leq {\mathcal K}/\dist^{2}_{g}(p,\partial M)$ (see c.f. {\bf Step I} in \cite{MR1806984}) that is all what we need here.}],
\ben
|Ric(p)|\leq {\mathcal K}/\displaystyle d^{2}(p)
\een
In particular for any $p\in {\mathcal A}_{r}(a,b)$, the Ricci curvature of the scaled metric $g_{r}$ is bounded as $|Ric_{r}(p)|_{r}\leq {\mathcal K}/a^{2}$. 

\vs
{\sc Norms and convergence of Riemannian manifolds.} 

\vs
- Given a tensor field $U$ (of any valence) on a region $\Omega$ of a manifold $(M,g)$, the $C^{i}_{g}$-norm of $U$ over $\Omega$ is defined as
\ben
\| U\|_{C^{i}_{g}(\Omega)}:=\sup_{p\in \Omega} \sum_{j=0}^{j=i} \big|\big(\nabla^{j} U\big)(p)\big|_{g}
\een
Of course $\|U\|_{C^{i}_{g}}\leq \|U\|_{C^{i+1}_{g}}$. The subindex $g$ will be suppressed when $\Omega$ is a region of the Euclidean three-space, namely we will write $C^{i}$. 

- All what we will need about convergence of smooth Riemannian manifolds will be restricted to the following definition (which is not the most general \cite{MR2243772}). Let $(\Omega_{m},g_{m})$ be a sequence of smooth, compact, connected three-manifolds with smooth boundary and let $(\Omega_{\infty},g_{\infty})$ be also smooth, compact, connected three-manifold with smooth boundary. Then, $(\Omega_{m},g_{m})$ converges to $(\Omega_{\infty},g_{\infty})$ in $C^{i}$, $i\geq 2$, if there are diffeomorphisms 
$\varphi_{m}:\Omega_{\infty}\rightarrow \Omega_{m}$ such that 
%\ben
$\big\| \varphi_{m}^{*}\, g_{m} - g_{\infty}\big\|_{C^{i}_{g_{\infty}}(\Omega_{\infty})}\rightarrow 0$
%\een
where $\varphi^{*}_{m} g_{m}$ is the pull-back of $g_{m}$ by $\varphi_{m}$. 
The definition is the same if we do not require compactness on the $\Omega_{m}$ and $\Omega_{\infty}$ but assume uniformly bounded diameters.
A sequence of smooth tensors $U_{m}$ converge to a smooth tensor $U_{\infty}$ in $C^{i}$, $i\geq 0$, if 
%\ben
$\big\| \varphi_{m}^{*}\, U_{m} - U_{\infty}\big\|_{C^{i}_{g_{\infty}}(\Omega_{\infty})}\rightarrow 0$
%\een
%

\vs
{\sc Elliptic estimates.} 

\vs
- Stationary solutions $(M;g,\omega,u)$ satisfy important regularity properties derived directly from the elliptic system (\ref{MEE}). 
\begin{Proposition} Let $(M;g,\omega,u)$ be a stationary solution and let $p\in M$. Suppose that all over $M$ we have
$R\leq R_{0}$ and that at $p$ we have $inj_{g}(p)\geq I_{0}>0$ and $u(p)+1/u(p)\leq \Lambda_{0}$, 
where $inj_{g}(p)$ is the injectivity radius at $p$. Then, for any $i\geq 0$ there are  constants $K^{i}_{1}(R_{0},I_{0})>0$ and $K^{i}_{2}(R_{0},I_{0},\Lambda_{0})>0$ such that  
\ben
\|Ric\|_{C^{i}_{g}(B)}
%(p,I_{0}/2))}
\leq K_{1}^{i}\qquad \text{and}\qquad \|\omega\|_{C^{i}_{g}(B)}+ \| u\|_{C^{i+1}_{g}(B)} \leq K_{2}^{i} 
\een
where $B=B_{g}(p,I_{0}/2)$. 
\end{Proposition} 

It is important to remark that the constants $K_{1}^{i}$ which bound the $C^{i}_{g}$-norm of $Ric$ do not depend on $\Lambda_{0}$. 
This is because any stationary solution can be scaled to the stationary solution $(g,\bar{\omega},\bar{u})=(g,\omega/u(p)^{2},u/u(p))$ which has $\bar{u}(p)=1$.  

%For this same reason, there is $K_{3}^{i}(R_{0},I_{0})>0$ such that the one form $\zeta=(\nabla {\mathcal E})/u^{2}=2du/u+i\omega/u^{2}$ (see {\sc The Ernst equation} above) which is invariant under the scaling before, has $C^{i}_{g}$-norm over $B_{g}(p,I_{0}/2)$ bounded by $K_{3}^{i}$. 

In the Part I no much will be required about elliptic estimates because most of the necessary estimates are already contained in the definition of WAF end. 
The following however is a simple application that will be used in the proof of Proposition \ref{PREL}.
Let $(\Omega_{m};g_{m},\omega_{m},u_{m})$ be a sequence of stationary solutions and suppose that $(\Omega_{m},g_{m})$ converges in $C^{2}$ to the flat annulus $(\AR(a,b),g_{\mathbb{R}^{3}})$. Then, $g_{m}$ converges to $g$ also in $C^{i}$ for any $i\geq 2$ and without the necessity of taking a subsequence. Moreover there are scalings $\bar{\omega}_{m}:= \lambda^{2}_{m}\omega_{m}$ and $\bar{u}_{m}:= \lambda_{m}u_{m}$ such that $\bar{\omega}_{m}$ and $\bar{u}_{m}$ converge in $C^{i-2}$ and $C^{i-1}$ to the one-form zero and a the constant function one respectively. In particular the scale invariant one-forms $\zeta_{m}=2(u_{m}\nabla u_{m}+i\omega_{m})/u_{m}^{2}$ converge in $C^{i-1}$ to the one-form zero.

\section{Weakly Asymptotically Flat ends}\label{WAFECHP}

\vs
\begin{Definition}\label{WAFD} A stationary end $E$ is weakly asymptotically flat (WAF) if it is strongly stationary and 
for every $i\geq 2$, $l\geq 1$ and divergent sequence $r_{m}\rightarrow \infty$, there is a sequence of open annuli $\Omega_{m}\subset E$ 
such that, 
\begin{enumerate}[labelindent=\parindent,leftmargin=*,label={\rm (W\arabic*)}]
\item ${\mathcal A}_{r_{m}}(1/2,2^{l})\subset \Omega_{m}$ for every $m$,

\item $(\Omega_{m},g_{r_{m}})$ converges in $C^{i}$ to  the flat annulus 
$\big({\mathcal A}_{\mathbb{R}^{3}}(1/2,2^{l}),g_{\mathbb{R}^{3}}\big)$,

\item The scaled distance functions $d_{r_{m}}$ (restricted to $\Omega_{m}$) converge in $C^{0}$ to the distance to the origin in $\mathbb{R}^{3}$ (restricted to ${\mathcal A}_{\mathbb{R}^{3}}(1/2,2^{l})).$

\item Every $\overline{\Omega}_{m}$ is a closed annulus and separates $\partial E$ from infinity, namely, $\partial E$ belongs to a bounded component of $E\setminus \overline{\Omega}_{m}$ for all $m$.

\end{enumerate}
\end{Definition}

\vs
The Figure \ref{Fig1} illustrates a WAF end together with some annuli $\Omega_{m}$. 

It can be shown that the condition (W4) is in fact redundant. Namely, if we define WAF ends exactly as in Definition \ref{WAFD} but removing (W4), then one such end would comply also with the Definition \ref{WAFD} including (W4).
The proof of that is not relevant to us and is skipped.

To be clear from the start, let us recall that the convergences in (W2) and (W3) express the existence of diffeomorphisms 
$
\varphi_{m}:\AR(1/2,2^{l})
\rightarrow \Omega_{m}
$
such that the metrics $\varphi_{m}^{*} g_{r_{m}}$ over $\AR(1/2,2^{l})$
converge in $C^{i}$ to the flat Euclidean metric and that, at the same time, the functions $d_{r_{m}} \circ \varphi_{m}$ over $\AR(1/2,2^{l})$ converge in $C^{0}$ to the distance function to the origin. 

We stress that in (W2) the convergence is for the entire sequence $(\Omega_{m},g_{r_{m}})$ and not just for a subsequence of it. Note also that by (W3) the sets $\varphi_{m}^{-1}({\mathcal A}_{r_{m}}(1/2,2^{l}))$ tend to cover the whole $\AR(1/2,2^{l})$, in the sense that for every $1/2>\varepsilon>0$ there is $m_{\varepsilon}$ such that if $m\geq m_{\varepsilon}$ then $\AR(1/2+\varepsilon,2^{l}-\varepsilon)\subset \varphi_{m}^{-1}({\mathcal A}_{r_{m}}(1/2,2^{l}))$   
The reason why we need the regions $\Omega_{m}$ in Definition \ref{WAFD} is essentially technical and not particularly important. As we show in Part II, strong stationary ends with cubic volume growth are WAF in the sense of Definition \ref{WAFD} and, to stick to the definition of convergence, the existence of the diffeomorphisms $\varphi_{m}$ is guaranteed only into the regions $\Omega_{m}$ (which cover tightly the metric annuli ${\mathcal A}_{r_{m}}(1/2,2^{l})$) but unfortunately not into the metric annuli ${\mathcal A}_{r_{m}}(1/2,2^{l})$.

\begin{figure}[h]
\centering
\includegraphics[width=9cm,height=7cm]{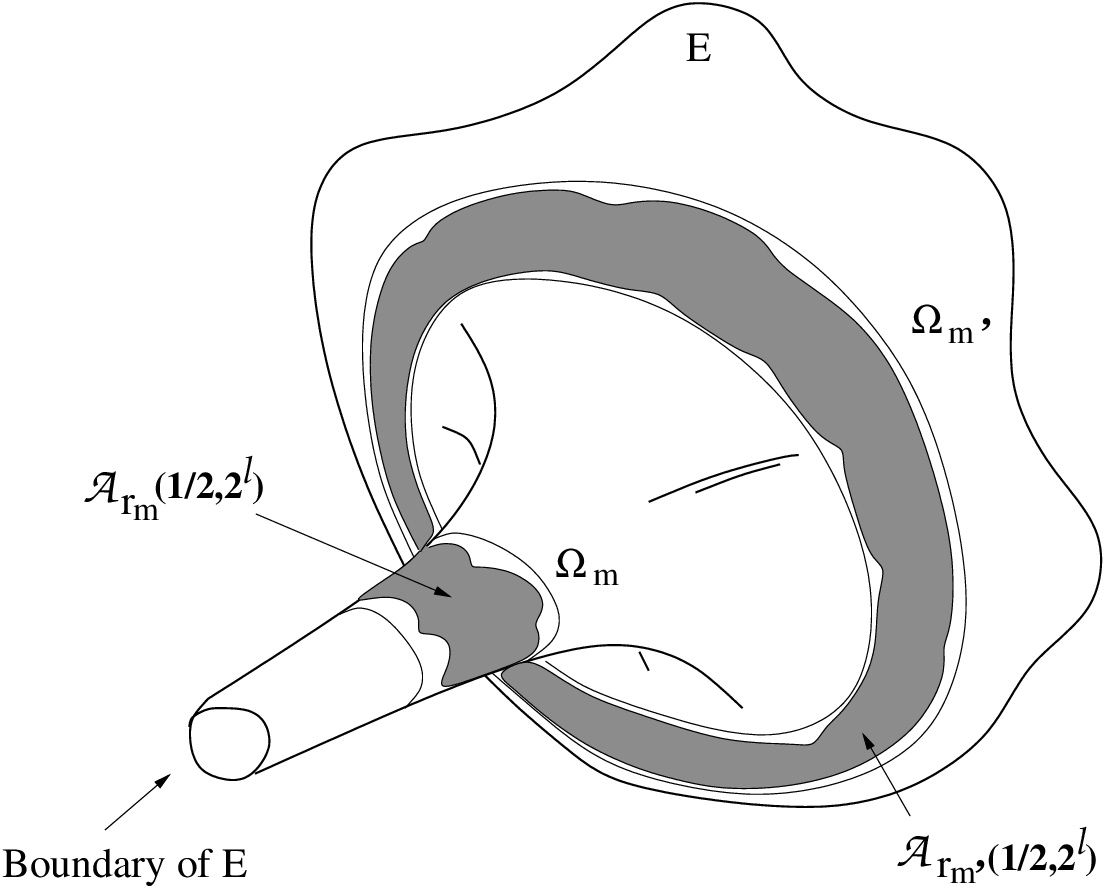}
\caption{Representation of a WAF end along with the annuli $\Omega_{m}$ and the metric annuli ${\mathcal A}_{r_{m}}(1/2,2^{l})$}
\label{Fig1}
\end{figure} 

WAF ends have a number of technical advantages allowing to prove the standard Schwarzschidian decay in a more comfortable setup. The next section discusses some basic properties of WAF ends.   

\subsection{Preliminaries of WAF ends}\label{PREWAFE}

The final goal of of this section is to prove Proposition \ref{COROL}. This proposition says that one can always restrict the domain of a given WAF end $E$ by removing from it an annulus (with one boundary component $\partial E$) 
to get an end simpler to handle for its nice geometric properties. This type of WAF ends, which are defined in Definition \ref{DEFEF}, we call $\varepsilon$-flat. 

The first proposition says that the curvature of WAF ends decays faster than quadratically. 

\begin{Proposition} Let $E$ be a WAF end. Then, for any divergent sequence of points $p_{m}$ (i.e. $d(p_{m})\rightarrow \infty$) we have
\ben
\lim_{m\rightarrow \infty} |Ric(p_{m})|d^{2}(p_{m})=0
\een
\end{Proposition}

\vs
\n Thus, if we write $|Ric(p)|=o(p)/d^{2}(p)$, then $o(p)\rightarrow 0$ when $d(p)\rightarrow \infty$. 

\begin{proof}[\bf Proof.] Let $p_{m}$ be any divergent sequence of points and let $r_{m}=d(p_{m})$. From the definition of WAF end (with $l=1$ and $i=2$), there is a sequence of regions $\Omega_{m}$, with $p_{m}\in {\mathcal A}_{r_{m}}(1/2,2)\subset \Omega_{m}$ for every $m$, such that $(\Omega_{m},g_{r_{m}})$ converges in $C^{2}$ to a flat annulus. Therefore $|Ric_{r_{m}}(p_{m})|_{r_{m}}\rightarrow 0$, or, equivalently, $|Ric(p_{m})| d^{2}(p_{m})\rightarrow 0$. 
\end{proof}

\vs
\begin{Proposition}\label{AAA} Let $E$ be a WAF end. Then, for every $\delta>0$ there is $r_{\delta}>0$ and a two-sphere $S_{\delta}$ separating $\partial E$ from infinity such that,
\begin{enumerate}[labelindent=\parindent,leftmargin=*,label={\rm (A\arabic*)}]
\item $1-\delta \leq \inf\big\{d_{r_{\delta}}(p)/p\in S_{\delta}\big\}\leq \sup \big\{d_{r_{\delta}}(p)/p\in S_{\delta}\big\}\leq 1+\delta$,

\item For all $p\in S_{\delta}$ we have $|\widehat{\Theta}_{r_{\delta}}(p)|_{r_{\delta}}\leq \delta$ and $|\theta_{r_{\delta}}(p)-2|\leq \delta$,

\item For all $p$ in the unbounded component $E_{S_{\delta}}$ of $E\setminus S_{\delta}$ we have 
\ben
|Ric_{r_{\delta}}(p)|_{r_{\delta}}\leq \frac{\delta}{(1+\dist_{r_{\delta}}(p,S_{\delta}))^{2}}
\een 
\end{enumerate}
\end{Proposition}

\vs
\n We recall from the notation in Section \ref{NOTATION} that in (A2) above, $\widehat{\Theta}_{r_{\delta}}$ is the traceless part of the second fundamental form $\Theta_{r_{\delta}}$ and $\theta_{r_{\delta}}$ is the mean curvature of the surface $S_{\delta}$ as a surface embedded in the Riemannian space $(E,g_{r_{\delta}})$ (hence the subindex $r_{\delta}$). To define $\Theta_{r_{\delta}}$ we assume a normal in $S_{\delta}$ inwards to $E_{S_{\delta}}$. 

\begin{proof}[\bf Proof.] Assume at the moment that $0<\delta<1/4$. This simplifies the algebra a little. We remove it at the end.
We start observing that, as the curvature of WAF stationary ends decays faster than quadratically, then: 

\vs
(i) {\it For every $0<\delta<1/4$ there is $\hat{r}_{\delta}>0$ such that for every $p$ with $d(p)\geq \hat{r}_{\delta}$ we have 
%$|Ric_{r}(p)|_{r}\leq \delta/4$ or, equivalently, 
$|Ric(p)|\leq \delta/4d^{2}(p)$ (the $4$ in the denominator is also for algebraic convenience).} 

\vs
\n Second, we observe that using Definition \ref{WAFD} (with $i=2$ and $l=1$) we can guarantee that: 

\vs
(ii) {\it For every $0<\delta<1/4$ there is $r_{\delta}\geq 2\hat{r}_{\delta}$ ($\hat{r}_{\delta}$ as in (i)), an open annulus $\Omega \supset {\mathcal A}_{r_{\delta}}(1/2,2)$, and a diffeomorphism $\varphi: \AR(1/2,2)\rightarrow \Omega$ such that the metric $\varphi^{*} g_{r_{\delta}}$ is sufficiently close in $C^{2}$ to the euclidean metric and the function $d_{r_{\delta}}\circ \varphi$ is sufficiently close in $C^{0}$ to the distance function to the origin in ${\mathbb{R}^{3}}$, that if we define $S_{\delta}:=\varphi(\partial B_{\mathbb{R}^{3}}(o,1))$ then (A1) and (A2) hold.} 

\vs
\n Thus, with this definition of $S_{\delta}$ we have already (A1) and (A2). That $S_{\delta}$ separates $\partial E$ from infinity is direct because, by (W4), $\overline{\Omega}$ is a closed annulus in $E$ separating $\partial E$ from infinity 
[\footnote{It is an exercise in topology to prove
that, as $E\sim \mathbb{R}^{3}\setminus B_{\mathbb{R}^{3}}(o,1)$ and as $\overline{\Omega}$ is a closed annulus embedded in $E$ separating $\partial E$ from infinity, then $E\setminus \Omega$ consists of two pieces, one which is a closed annulus containing $\partial E$ and the other which is diffeomorphic to $\mathbb{R}^{3}\setminus B_{\mathbb{R}^{3}}(o,1)$.}]. 
It remains to prove (A3). We do that now. For any $p\in E_{S_{\delta}}$ (the unbounded component of $E\setminus E_{S_{\delta}}$) we have
\be\label{OPOBOU}
d(p)\geq \dist(p,S_{\delta})+\dist(S_{\delta},\partial E)
\ee
To see this consider a geodesic segment from $p$ to $\partial E$ whose length realizes $d(p)$. Such segment must intersect $S_{\delta}$  in (at least) one point $q$. Hence $d(p)=\dist(p,q)+\dist(q,\partial E)\geq \dist(p,S_{\delta})+\dist(S_{\delta},\partial E)$ as wished. Dividing the expression (\ref{OPOBOU}) by $r_{\delta}$ we obtain
%\ben
$
d_{r_{\delta}}(p)\geq \dist_{r_{\delta}}(p,S_{\delta})+\dist_{r_{\delta}}(S_{\delta},\partial E)\
$
%\een
and from this and (A1) we get 
\be\label{LLG}
d_{r_{\delta}}(p)\geq \dist_{r_{\delta}}(p,S_{\delta})+1-\delta. 
\ee
Observe then that as $0<\delta<1/4$ we have
$d_{r_{\delta}}(p)\geq 1-\delta>1/2$ and thus $d(p)>r_{\delta}/2$. But as in (ii) we assumed $r_{\delta}\geq 2\hat{r}_{\delta}$ then $d(p)\geq\hat{r}_{\delta}$ and we can use (i). Use now (\ref{LLG}) in $|Ric_{r_{\delta}}(p)|_{r_{\delta}}\leq \delta/4d^{2}_{\delta}(p)$, which we obtain from (i) by scaling $|Ric(p)|\leq \delta/4d^{2}(p)$, to deduce
\ben
|Ric_{r_{\delta}}(p)|_{r_{\delta}}\leq \frac{\delta/4}{(d_{r_{\delta}}(p)+1-\delta)^{2}}\leq \frac{\delta}{(\dist_{r_{\delta}}(p,S_{\delta})+1)^{2}}
\een
where the second inequality holds as long as $0<\delta<1/4$ as we are assuming. This shows (A3) as wished.

When $\delta\geq 1/4$ then define $r_{\delta}:=r_{1/8}$ and $S_{\delta}:=S_{1/8}$. With this definition (A1)-(A3) are immediate because $\delta>1/8$ (for instance $|\widehat{\Theta}_{r_{\delta}}(p)|_{r_{\delta}}=|\widehat{\Theta}_{r_{1/8}}(p)|_{r_{1/8}}\leq 1/8\leq \delta$).\end{proof}

\vs
Let $(M,g)$ be a complete Riemannian manifold with $M$ diffeomorphic to $\mathbb{R}^{3}$ minus a ball. At least for a $t_{*}>0$ small, the equidistant sets ${\mathcal S}(t):=\{p\in M, \dist(p,\partial M)=t\}$, where $t\leq t_{*}$, are embedded spheres. For any $p\in {\mathcal S}(t)$ with $0\leq t\leq t_{*}$ denote by $\theta(p)$ and by $\widehat{\Theta}(p)$ to the mean curvature and traceless part of the second fundamental form of ${\mathcal S}(t)$ at $p$ (in the direction of increasing $t$) respectively.   
We will use this notation (c.f. {\sc second fundamental form} in Section \ref{NOTATION}) in the statement of the next proposition which will be used later in conjunction with Proposition \ref{AAA} to deduce Proposition \ref{COROL}.
In the item (B2) inside the statement below we let $\const>0$ be any numeric constant such that $|Rm|\leq 4\const |Ric|$ (recall that in dimension three the Ricci curvature determines the Riemann curvature by an algebraic formula). The constant $\const$ plays some algebraic role later but is not particularly important. 

\vs
\begin{Proposition} \label{LLI} For any $\varepsilon>0$ there is $\delta>0$ such that if a complete smooth Riemannian three-manifold    
$(M,g)$, with $M$ diffeomorphic to $\mathbb{R}^{3}$ minus a ball, satisfies 
\begin{enumerate}[labelindent=\parindent,leftmargin=*,label={\rm (B\arabic*)}]
\item $|\theta(p)-2|\leq \delta$ and $|\widehat{\Theta}(p)|\leq \delta$ for all $p\in \partial M$, and,
\item $|Ric(p)|\leq \delta/\big(\const(1+\dist(p,\partial M))^{2}\big)$ for all $p\in M$,
\end{enumerate}
then, 
\be\label{LUB}
|\theta(p) -\frac{2}{1+t}|\leq \frac{\varepsilon}{1+t}\qquad \text{and}\qquad |\widehat{\Theta}(p)|\leq \frac{\varepsilon}{1+t}
\ee
as long as $p\in S(t)=\{p/\dist(p,\partial M)=t\}$ and $t\leq t_{*}$, where $t_{*}>0$ is such that for any $t\in [0,t_{*}]$ the equidistant sets $S(t)$ are embedded surfaces.
\end{Proposition}

\vs
That $M$ is diffeomorphic to $\mathbb{R}^{3}$ minus a ball will play not role in the proof. Despite of this we will keep this assumption for several expository reasons.

\n \begin{proof}[\bf Proof] Let ${\mathcal F}$ be the congruence of geodesics in $M$ emanating perpendicularly from $\partial M$. We will denote such geodesics as $\gamma_{q}(t)$, $q\in \partial M$. Observe that if $\gamma_{q}(t)\in S(t')$ then $t=t'$. For this reason we would be done if we prove that over every geodesic $\gamma_{q}(t)$ we have $|\theta(\gamma_{q}(t))-2/(1+t)|\leq \varepsilon/(1+t)$ and $|\widehat{\Theta}(\gamma_{q}(t))|\leq \varepsilon/(1+t)$. This is exactly how we will proceed and to do so we will study a couple of differential inequalities that $\theta$ and $|\widehat{\Theta}|$ satisfy along the geodesics $\gamma_{q}(t)$ (equations (\ref{INEI}) and (\ref{INEII})) to obtain upper and lower bounds for $\theta(\gamma_{q}(t))$ and $|\widehat{\Theta}(\gamma_{q}(t))|$ essentially equivalent to (\ref{LUB}).  
We move first to deduce these differential inequalities ((\ref{INEI}) and (\ref{INEII})). 

In the computations below $h(t)$ denotes the Riemannian-metric induced over $S(t)$. Recall that the Lie derivative of $\Theta$ along the velocity field\footnote{That is, $n(\gamma_{q}(t))=\gamma'_{q}(t)$.} $n$ of the congruence ${\mathcal F}$ is given by 
\be\label{SFFEE}
\Theta'=\Theta\circ\Theta -Rm(n,n)
\ee

\n where $(\Theta\circ \Theta)_{ab}:=\Theta_{a}^{\ c}\Theta_{cb}$ and where $Rm(n,n)$ is the symmetric two-form $Rm(n,n)(v,w):=Rm(v,n,w,n)$. The contraction of this equation gives, as is well known, the focusing equation
\be\label{FOCUS}
\theta'=-\frac{\displaystyle \theta^{2}}{2}-|\widehat{\Theta}|^{2}-Ric(n,n)
\ee
This equation and the ones below are to be evaluated along the geodesics $\gamma_{q}$ but we will omit writing explicitly this dependence for notational convenience. Also, for later advantage we will use $\bar{t}:=1+t$ instead of $t$ as the parameter of the geodesics $\gamma_{q}$. 
We claim now that from (\ref{SFFEE}) and (\ref{FOCUS}) we obtain the equation 
\be\label{DELTF}
|\widehat{\Theta}|^{2}{'}=-2\langle\widehat{\Theta},\widehat{\Theta}\circ\widehat{\Theta}\rangle-2\theta|\widehat{\Theta}|^{2}+2\langle\widehat{\Theta},\widehat{Rm}(n,n)\rangle
\ee
where $\widehat{Rm}(n,n)$ is the traceless part of $Rm(n,n)$ and where $\langle\ ,\ \rangle$ is the inner product among symmetric two-tensors defined by $\langle U,V\rangle=U_{ab}V_{cd}h^{ac}h^{bc}$. To see this we compute
\begin{align*}
\widehat{\Theta}'&=\Theta'-\frac{\displaystyle \theta'}{2}h-\theta\Theta\\
&=\Theta\circ\Theta -Rm(n,n)+\big(\frac{\theta^{2}}{2}+|\widehat{\Theta}|^{2}+Ric(n,n)\big)\frac{h}{2}-\theta\Theta,\\
&=\widehat{\Theta}\circ\widehat{\Theta}+\frac{|\widehat{\Theta}|^{2}}{2}h-{\widehat{Rm}}(n,n)
\end{align*}
\n where to obtain the first equality we use that $h'=2\Theta$. Therefore
\begin{align*}
|\widehat{\Theta}|^{2}{'}&=2\langle \widehat{\Theta},\widehat{\Theta}'\rangle-4\langle\Theta,\widehat{\Theta}\circ \widehat{\Theta}\rangle\\
&=2\langle\widehat{\Theta},\widehat{\Theta}\circ \widehat{\Theta}+\frac{|\widehat{\Theta}|^{2}}{2}h-\widehat{Rm}(n,n)\rangle-4\langle\Theta,\widehat{\Theta}\circ\widehat{\Theta}\rangle\\
&=-2\langle\widehat{\Theta},\widehat{\Theta}\circ \widehat{\Theta}\rangle-2\langle\widehat{Rm}(n,n),\widehat{\Theta}\rangle-2\theta |\widehat{\Theta}|^{2}
\end{align*}
as wished [\footnote{To deduce the first equality in this calculation use that $\langle \Theta,\Theta\rangle = \Theta_{ab}\Theta_{cd}h^{ac}h^{bd}$ and that $(h^{ab})'=-(h_{cd}')h^{ca}h^{bd}=-2\Theta^{ab}$.}]. 
\n The equations (\ref{DELTF}) and (\ref{FOCUS}) give the inequalities
\begin{align}
& |\widehat{\Theta}|^{2}{'}\leq 2|\widehat{\Theta}|^{3}-2\theta|\widehat{\Theta}|^{2}+2|\widehat{\Theta}||\widehat{Rm}(n,n)|,\label{INEI}\\
& -\frac{\theta^{2}}{2}-|\widehat{\Theta}|^{2}-|Ric(n,n)|\leq \theta'\leq -\frac{\theta^{2}}{2}\label{INEII}.
\end{align}

\n which hold along every geodesic $\gamma_{q}(\bar{t})$. We will analyze them in what follows under the hypothesis (B1) and (B2). We will then adjust $\delta$ to satisfy (\ref{LUB}) for the given $\varepsilon$. As a matter of fact the first adjustment of $\delta$ is $\delta\leq \min \{\varepsilon, 1/32\}$ that will be assumed from now on. The reason for this will become clear later.

One can easily get a first conclusion just analyzing the second inequality in (\ref{INEII}). Indeed from it, $\delta\leq \epsilon$ and (B1), we deduce that $\theta(\bar{t})$ satisfies
\ben
\left\{
\begin{array}{l}
\theta'\leq -\frac{\displaystyle \theta^{2}}{\displaystyle 2},\vs\\
\theta(1)\leq 2+\varepsilon
\end{array}
\right.
\een
One can then directly compare $\theta(\bar{t})$ to $x(\bar{t})=(2+\varepsilon)/\bar{t}$ because $x(\bar{t})$ satisfies 
\ben
\left\{
\begin{array}{l}
x'\geq -\frac{\displaystyle x^{2}}{\displaystyle 2},\vs\\
x(1)=2+\varepsilon
\end{array}
\right.
\een
to conclude, from a standard ODE analysis, that $x(\bar{t})$ is an upper barrier to $\theta(\bar{t})$, that is $\theta(\bar{t})\leq (2+\varepsilon)/\bar{t}$. Hence to show (\ref{LUB}) it remains to prove that one can adjust $\delta$ further to have also 
\be\label{DESIRE}
\frac{2-\varepsilon}{\bar{t}}\leq \theta(\bar{t})\qquad \text{and}\qquad |\widehat{\Theta}(\bar{t})|\leq \frac{\varepsilon}{\bar{t}}. 
\ee
The proof of these two bounds is simultaneous and is done at the end of the discussion below.

\vs
$\circ$ {\it Assume that for any $\bar{t}\in [1,\bar{t}_{a}]$ we have $\theta(\bar{t})\geq 3/2\bar{t}$}. In the analysis below we will be restricted to this interval of $\bar{t}$. Use now this assumption and (B2) in (\ref{INEI}) to get 
\be\label{INEIE}
|\widehat{\Theta}|^{2}{'}\leq 2|\widehat{\Theta}|^{3}-\frac{3}{\bar{t}}|\widehat{\Theta}|^{2}+\frac{2\delta}{\bar{t}^{2}}|\widehat{\Theta}|
\ee
where we to deduce the last term on the r.h.s use in addition that 
\ben
|\widehat{Rm}(n,n)|=|Rm(n,n)-\big({\rm tr}_{h}Rm(n,n)\big) \frac{h}{3}|\leq 4|Rm(n,n)|\leq 4\const |Ric(n,n)|\leq \delta/\bar{t}^{2}
\een
 by (B2) and by how the constant $\const$ was define (see the definition before the statement of the proposition). Consider then the first order ODE
\be\label{TXI}
(x^{2}){'}=2x^{3}-\frac{3}{\bar{t}}x^{2}+\frac{2\hat{\delta}}{\bar{t}^{2}}x
\ee
which is obtained by making $|\widehat{\Theta}|=x=x(\bar{t})$ in (\ref{INEIE}), then changing the inequality by an equality and finally changing $\delta$ by $\hat{\delta}$ which we assume to satisfy $\delta\leq \hat{\delta}\leq \min\{\varepsilon,1/32\}$ [\footnote{To obtain an ODE of the form $x'=F(x,\bar{t})$ use $(x^{2})'=2xx'$ and then divide by $2x$ in (\ref{TXI}).}]. 
The reason why we consider this ODE is the following: if $x(\bar{t})$ is a positive solution to (\ref{TXI}) such that $x(1)\geq |\widehat{\Theta}(1)|$ then $x(\bar{t})\geq |\widehat{\Theta}(\bar{t})|$ as long as they are defined. 
We look now for a solution to (\ref{TXI}) of the form $x=a/\bar{t}$, where $a$ is a constant. Substituting $x=a/\bar{t}$ in (\ref{TXI}) we obtain
\be\label{ODEI}
-2a^{2}\frac{1}{\bar{t}^{3}}=(2a^{3}-3a^{2}+2\hat{\delta} a)\frac{1}{\bar{t}^{3}}
\ee
Canceling the factor $a/\bar{t}^{3}$ and solving for $a$ we obtain that $x=a/\bar{t}$ ($a\neq 0$) is a solution to (\ref{TXI}) iff $a$ satisfies  
\ben
a=\frac{1}{4}(1\pm\sqrt{1-16\hat{\delta}})
\een
For us it will be important only the solution corresponding to the smaller $a$, namely $a=a_{-}:=(1-\sqrt{1-16\hat{\delta}})/4$.
Note for later reference that $a_{-}>0$ and that $a_{-}\rightarrow 0$ as $\hat{\delta}\rightarrow 0$. 

With the solution $x=a_{-}/\bar{t}$ at hand we can obtain the following {\it First Conclusion}

\vs
$\bullet$ {\it If $|\widehat{\Theta}(1)|\leq a_{-}$ and $\theta(1)\geq 3/2$, then $|\widehat{\Theta}(\bar{t})|\leq a_{-}/\bar{t}$ as long as $\theta(\bar{t})\geq 3/2\bar{t}$.}

\vs 
$\circ$ {\it Assume now that for any $\bar{t}\in [1,\bar{t}_{b}]$ we have $|\widehat{\Theta}(\bar{t})|\leq a_{-}/\bar{t}$ (where $a_{-}(\hat{\delta})$ is as before)}. In the analysis below we will be restricted to this interval of $\bar{t}$. Use this assumption and (B2) in the inequality (\ref{INEII}) to get
\be\label{INEIIE}
\theta'\geq -\frac{\theta^{2}}{2}-\frac{(a_{-}^{2}+\delta)}{\bar{t}^{2}}.
\ee
Consider then the first order ODE
\be\label{ODEII}
x'=-\frac{x^{2}}{2}-\frac{(a_{-}^{2}+\hat{\delta})}{\bar{t}^{2}}
\ee

\n which is obtained by making $\theta=x$ in (\ref{INEIIE}), then changing the inequality by an equality and finally changing $\delta$ by $\hat{\delta}$. Again, the reason why we consider this ODE is the following: if $x(\bar{t})$ is a solution to (\ref{ODEII}) such that $x(1)\leq \theta(1)$ then
$x(\bar{t})\leq \theta(\bar{t})$ as long as they are defined. We look now for a solution to (\ref{ODEII}) of the form $x=b/\bar{t}$ where $b$ is a constant. Substituting $x=b/\bar{t}$ in (\ref{ODEII}) we obtain
\ben
-b\frac{1}{\bar{t}^{2}}=(-\frac{b^{2}}{2}-(a_{-}^{2}+\hat{\delta}))\frac{1}{\bar{t}^{2}}
\een
Canceling the factor $1/\bar{t}^{2}$ and solving for $b$ we deduce that $b/\bar{t}$ is a solution to (\ref{ODEII}) 
iff $b$ satisfies
\ben
b=1\pm\sqrt{1-2(a_{-}^{2}+\hat{\delta})}.
\een
The solution $b_{+}$ will be the only important. Note for reference below that $b_{+}<2$ and $b_{+}\rightarrow 2$ as $\hat{\delta}\rightarrow 0$. 

With the solution $x=b/\bar{t}$ at hand we can obtain the following {\it Second Conclusion}

\vs
$\bullet$ {\it If $\theta(1)\geq b_{+}$ and $|\widehat{\Theta}(1)|\leq a_{-}$, then $\theta(\bar{t})\geq b_{+}/\bar{t}$ as long as $|\widehat{\Theta}(\bar{t})|\leq a_{-}/\bar{t}$.}

\vs
We proceed to combine both {\it Conclusions} to adjust finally $\delta$ to satisfy (\ref{DESIRE}). 
Chose $\hat{\delta}$ smaller than $\min\{\varepsilon,1/32\}$ if necessary to have
\begin{align}
a_{-}(\hat{\delta})\leq \varepsilon\qquad\text{and}\qquad b_{+}(\hat{\delta})>\max\big\{2-\varepsilon,\frac{3}{2}\big\}
\end{align} 
Then we make the choice 
\ben
\delta=\min\big\{\hat{\delta},a_{-}(\hat{\delta}),2-b_{+}(\hat{\delta})\big\}
\een
Observe that with this choice of $\delta$, the hypothesis (B1) implies  
\ben
\theta(1)\geq 2-\delta\geq b_{+}>2-\varepsilon\qquad\text{and}\qquad |\widehat{\Theta}(1)|\leq \delta\leq a_{-}(\hat{\delta})\leq \varepsilon
\een
We claim that for all $\bar{t}\in [1,1+t_{*}]$ we have 
\be\label{TBP}
\theta(\bar{t})\geq \frac{b_{+}}{\bar{t}}\qquad\text{and}\qquad |\widehat{\Theta}(\bar{t})|\leq \frac{a_{-}}{\bar{t}}
\ee
that would imply (\ref{DESIRE}) because $b_{+}>2-\varepsilon$ and $a_{-}\leq \varepsilon$. Hence we would be done after proving the claim. We do that below. 

Suppose the claim is false and let $\bar{t}_{F}<t_{*}+1$ be the last time for which both inequalities in (\ref{TBP}) hold. 

{\it Case 1}. Suppose there are times $\bar{t}$ greater than $\bar{t}_{F}$ but arbitrarily close to it for which $|\widehat{\Theta}(\bar{t})|>a_{-}/\bar{t}$. Let $\bar{t}'_{F}>\bar{t}_{F}$ close enough to $\bar{t}_{F}$ that on $[\bar{t}_{F},\bar{t}'_{F}]$ we still have $\theta(\bar{t})>3/2\bar{t}$. The according to the {\it First Conclusion} there must be a time $\bar{t}_{1}$ before $\bar{t}_{F}'$ for which $\theta(\bar{t}_{1})<3/2\bar{t}_{1}$ which is not possible. 

{\it Case 2}. Suppose instead that there is $\bar{t}_{F}'$ such that on $[\bar{t}_{F},\bar{t}'_{F}]$ we still have $|\widehat{\Theta}(\bar{t})|\leq a_{-}/\bar{t}$ but that there are times $\bar{t}$ greater than $\bar{t}_{F}$ but arbitrarily close to it for which $\theta(\bar{t})<b_{+}/\bar{t}$. Then according to the {\it Second Conclusion} there must be a time $\bar{t}_{1}$ less than $\bar{t}'_{F}$ for which $|\widehat{\Theta}(\bar{t}_{1})|>a_{-}/\bar{t}$ which is not possible.
\end{proof}

\vs
\begin{Proposition}\label{COROL} Let $E$ be a WAF end. Then for every $0<\varepsilon<1/2$ there is $r_{\varepsilon}>0$ and an embedded two sphere $S_{\varepsilon}$ separating $\partial E$ from infinity, such that
\begin{enumerate}[labelindent=\parindent,leftmargin=*,label={\rm (U\arabic*)}]
\item On the unbounded component $E_{S_{\varepsilon}}$ of $E\setminus S_{\varepsilon}$ the distance function $\dist_{r_{\varepsilon}}(p,S_{\varepsilon})$ is smooth and every level set $S(t)=\{p,\dist_{r_{\varepsilon}}(p,S_{\varepsilon})=t\}$ is an embedded sphere, and,
\item For every $p\in S(t)$ we have 
\begin{align*}
& |\widehat{\Theta}_{r_{\varepsilon}}(p)|_{r_{\varepsilon}}\leq \frac{\varepsilon}{1+t},\qquad |\theta_{r_{\varepsilon}}(p)-\frac{2}{1+t}|\leq \frac{\varepsilon}{1+t},
\quad \text{and}\quad 
|Ric_{r_{\varepsilon}}(p)|_{r_{\varepsilon}}\leq \frac{\varepsilon}{(1+t)^{2}}
\end{align*}
\end{enumerate}
\end{Proposition}

The condition $\varepsilon<1/2$ is not relevant for the proof but helps for algebraic reasons. It is also included for compatibility with Definition \ref{DEFEF}, which is motivated by Proposition \ref{COROL} and where the condition is required.

\begin{proof}[\bf Proof.] For the given $\varepsilon$ denote by $\delta^{\varepsilon}$ the delta provided by Proposition \ref{LLI}. Then, by Proposition \ref{AAA} one can find for every $\delta\leq \min\{\delta^{\varepsilon}/\const,\varepsilon\}$ a sphere $S_{\delta}$ and a $r_{\delta}$ such that (A1)-(A3) hold. Here $\const$ is the numeric constant defined before the statement of Proposition \ref{LLI}. Hence (B1) and (B2) hold too if we define the manifold $(M,g)$ as $(E_{S_{\delta}},g_{r_{\delta}})$ and we can use Proposition \ref{LLI}. If we define $r_{\varepsilon}$ (the one claimed by the proposition) as $r_{\varepsilon}:=r_{\delta}$, we conclude using (\ref{LUB}) that (U2) will be valid as long as $t\in [0,t^{*}_{\delta})$ where $t^{*}_{\delta}$ is the supremum of the $t_{*}>0$ such that for every $t\in [0,t_{*}]$ the set $S(t)=\{p\in E_{S_{\delta}},\dist_{r_{\delta}}(p,S_{\delta})=t\}$ is an embedded sphere. 

Observe that this construction is valid for every $\delta\leq \{\delta^{\varepsilon},\varepsilon\}$ and also that given $\delta$
%, by a simple inspection of the proof of Proposition \ref{AAA}, 
we can chose $r_{\delta}$ (and therefore $r_{\varepsilon}$) larger than any given number (by a simple inspection of the proof of Proposition \ref{AAA}).  
To prove (U1) it is enough to show that we can chose $\delta$ small enough and then $r_{\delta}$ big enough to have $t_{\delta}^{*}=\infty$. To this purpose it suffices to show that for any sequences $\delta^{i}\downarrow 0$ (with $\delta^{i=0}\leq \min\{\delta^{\varepsilon},\varepsilon\}$)  
and $r_{\delta^{i}}\uparrow \infty$ there is some $i$ for which $t^{*}_{\delta^{i}}=\infty$. Consider then one such pair of sequences and suppose, arguing by contradiction, that $t^{*}_{\delta^{i}}$ is finite for every $i$. We will show that this is impossible.
From the construction of the spheres $S_{\delta^{i}}$ (c.f. (ii) inside the proof of Proposition \ref{AAA}) 
there is, for every $i$, an annulus $\Omega_{i}\supset {\mathcal A}_{r_{\delta^{i}}}(1/2,2)$ together with a diffeomorphism $\varphi_{i}:\AR(1/2,2)\rightarrow \Omega_{i}$, both provided by the definition of WAF end,
and with $\varphi_{i}(\partial B_{\mathbb{R}^{3}}(o,1))=S_{\delta^{i}}$. As $r_{\delta^{i}}\rightarrow \infty$, then (also from the definition of WAF end) the metrics $\varphi^{*}_{i} g_{r_{\delta^{i}}}$ converge in $C^{2}$ to the Euclidean metric and the functions $d_{r_{\delta^{i}}}\circ \varphi_{i}$ converge in $C^{0}$ to the distance function to the origin.    
But in $\mathbb{R}^{3}$ the equidistant sets $\{x\in \AR(1,2), \dist_{\mathbb{R}^{3}}(x,\partial B_{\mathbb{R}^{3}}(o,1))=t\}$ are obviously equal to the embedded spheres $\partial B_{\mathbb{R}^{3}}(o,1+t)$ for all $t\in [0,1)$. Therefore, by continuity, there is $i_{0}$ such that if $i\geq i_{0}$ then $t^{*}_{\delta^{i}}\geq 1/2$. We assume then from now on and without loss of generality that $t_{\delta^{i}}^{*}\geq 1/2$ for every $i\geq 0$. 

By (U2) and (for every $i$) the mean curvature of the spheres $S(t)$ remains finite for every $t\in [0,t^{*}_{\delta^{i}}]$.
Thus (for every $i$) the surfaces $S(t)$ are embedded when $t<t^{*}_{\delta^{i}}$ but just at $t=t^{*}_{\delta^{i}}$ the surface $S(t^{*}_{\delta^{i}})$ is only immersed. We conclude that (for every $i$) there is at least a point $p^{*}_{i}$ in $S(t^{*}_{\delta^{i}})$ of self tangency of $S(t^{*}_{\delta^{i}})$ (see Figure \ref{Fig2}). As the surfaces $S(t)$ are equidistant to $S_{\delta^{i}}$ it is deduced that
at $p^{*}_{i}$ and at $t=t^{*}_{\delta^{i}}$ there arrive two geodesic segments $\gamma^{i}_{1}(t)$ and $\gamma^{i}_{2}(t)$ (i.e.  $\gamma^{i}_{1}(t^{*}_{\delta^{i}})=\gamma^{i}_{2}(t^{*}_{\delta^{i}})$) that
start at $S_{\delta^{i}}$ when $t=0$. Moreover the geodesics $\gamma_{1}^{i}(t)$ and $\gamma^{i}_{2}(t)$ cross $S(t')$ only at $t=t'$ and    do so perpendicularly. For these reasons when they reach $p^{*}_{i}$ they do with opposite velocities, that is $\gamma_{1}'(t^{*}_{\delta^{i}})=-\gamma_{2}'(t^{*}_{\delta^{i}})$, and moreover we have
%\gamma^{i}_{1}(t^{*}_{\delta^{i}})=\gamma^{i}_{2}(t^{*}_{\delta^{i}})=p^{*}_{i},\qquad 
$\dist_{r_{\delta}}(\gamma_{1}(t),S_{\delta})=\dist_{r_{\delta}}(\gamma_{2}(t),S_{\delta})=t$ for all $t\in [0,t^{*}_{\delta^{i}}]$. 
These geodesic segments are depicted in Figure \ref{Fig2}. Note that one can form a larger geodesic segment, denoted here by $\gamma^{i}$, simply by concatenating $\gamma_{1}^{i}$ and $\gamma_{2}^{i}$ at $p^{*}_{i}$.
The distance $r^{*}_{i}:=d(p^{*}_{i})$ will be important below. Also it will be useful to parametrize $\gamma^{i}$ with the $g_{r^{*}_{i}}$-(signed) arc length $\tau^{i}$ starting from $p^{*}_{i}$ (in one of the two directions). As $t$ (which is a $g_{r_{\delta^{i}}}$-arc length) ranges in $[0,t^{*}_{\delta^{i}}]$ then $\tau^{i}$ ranges in $[-(r^{*}_{i}/r_{\delta^{i}}) t^{*}_{\delta^{i}}, (r^{*}_{i}/r_{\delta^{i}})t^{*}_{\delta^{i}}]$. We show later that $r^{*}_{i}/r_{\delta^{i}}\geq 1$. From this and $t^{*}_{\delta^{i}}\geq 1/2$ we deduce that $\tau^{i}$ ranges at least in $[-1/2,1/2]$.

\begin{figure}[h]
\centering
\includegraphics[width=7cm,height=9cm]{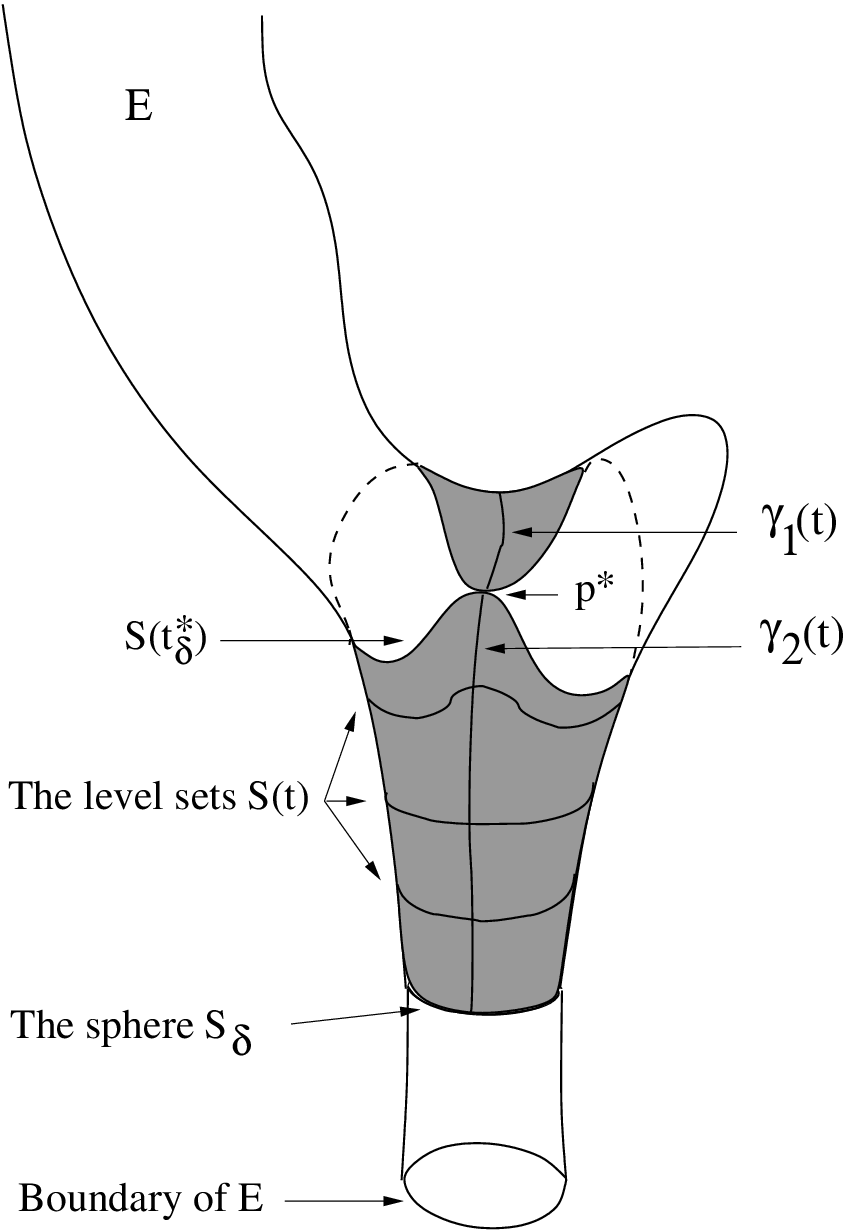}
\caption{Representation of the geometric elements in the proof by contradiction of (U1) in Proposition \ref{COROL}. For simplicity we do not include the index ``i".}
\label{Fig2}
\end{figure} 

To reach a contradiction we will use the following inequality \vs
\be\label{DISTDIFII}
\big|\, d_{r^{*}_{i}}(\gamma^{i}(\tau^{i}))-d_{r^{*}_{i}}(\gamma^{i}(-\tau^{i}))\, \big|\leq 2\delta^{i} \frac{r_{\delta^{i}}}{r^{*}_{i}}
\ee 
We prove this also later but for the moment and to avoid much disruption we proceed to use it. By the definition of WAF end we can consider a sequence of annuli $\Omega_{i}\supset {\mathcal A}_{r^{*}_{i}}(1/2,2)$ together with the sequence of diffeomorphisms $\varphi_{i}:\AR(1/2,2)\rightarrow \Omega_{i}$ such that $\varphi_{i}^{*}g_{r^{*}_{i}}$ converges in $C^{2}$ to the Euclidean metric and such that $d_{r^{*}_{i}}\circ \varphi_{i}$ converges in $C^{0}$ to the distance function to the origin. Taking a subsequence if necessary, the pull back of the geodesics segments $\gamma^{i}(\tau^{i})$ converge to a geodesic segment $\gamma(\tau)$ in $\mathbb{R}^{3}$ (hence a straight segment) with $|\gamma(\tau=0)|=1$ because $1=d_{r^{*}_{i}}(p^{*}_{i}=\gamma^{i}(0))\rightarrow |\gamma(0)|$
(here $|\gamma(\tau)|$ is the norm of $\gamma(\tau)$ as a point in $\mathbb{R}^{3}$, hence the Euclidean distance from $\gamma(\tau)$ to the origin). Moreover 
\ben
\big|\, |\gamma(\tau)|-|\gamma(-\tau)|\,\big|=0,\quad \text{for all } \tau\in [-\frac{1}{2},\frac{1}{2}]
\een
by taking the limit of (\ref{DISTDIFII}) (observe when taking the limit that $r_{\delta^{i}}/r^{*}_{i}\leq 1$). This equality is clearly impossible if $|\gamma(0)|=1$ and we reach a contradiction. 

\vs
To finish the proof it remains to prove (\ref{DISTDIFII}) and also $r^{*}_{i}/r_{\delta^{i}}\geq 1$. We show first (\ref{DISTDIFII}). To simplify the notation make below $\delta=\delta^{i}$, $r_{\delta}=r_{\delta^{i}}$, $r^{*}_{i}=r^{*}$, $t^{*}_{\delta^{i}}=t^{*}_{\delta}$, $\gamma_{1}=\gamma^{i}_{1}$ and $\gamma_{2}=\gamma^{i}_{2}$. In this notation (\ref{DISTDIFII}) is equivalent to
\be\label{DISTDIFIII}
\big|\, d_{r^{*}}(\gamma_{1}(t)) - d_{r^{*}}(\gamma_{2}(t))\, \big|\leq 2\delta \frac{\displaystyle r_{\delta}}{\displaystyle r^{*}}
\ee
We claim that this follows from proving, for any $p\in E_{S_{\delta}}$, the inequality
\be\label{DISTF}
\dist_{r_{\delta}}(p,S_{\delta})+1-\delta\leq d_{r_{\delta}}(p)\leq \dist_{r_{\delta}}(p,S_{\delta})+1+\delta
\ee
Indeed, if in it we make $p=\gamma_{j}(t)$, $j=1,2$ and then use that $\dist_{r_{\delta}}(\gamma_{j}(t),S_{\delta})=t$ for $j=1,2$, we get
\ben
t+1-\delta\leq d_{r_{\delta}}(\gamma_{j}(t))\leq t+1+\delta,\qquad \text{for } j=1,2
\een
and therefore
\ben
(t+1-\delta)\frac{r_{\delta}}{r^{*}}\leq d_{r^{*}}(\gamma_{j}(t))\leq (t+1+\delta)\frac{r_{\delta}}{r^{*}},\qquad \text{for } j=1,2
\een
from which (\ref{DISTDIFIII}) directly follows. We now deduce (\ref{DISTF}). Let $\Upsilon$ be a geodesic segment from $p$ to $\partial E$ whose length realizes the distance $d_{r_{\delta}}(p)$. $\Upsilon$ intersects $S_{\delta}$ in a point that we denote by $q$. Then we have
\ben
d_{r_{\delta}}(p)=\dist_{r_{\delta}}(p,q)+d_{r_{\delta}}(q) \geq \dist_{r_{\delta}}(p,S_{\delta})+1-\delta
\een
because $\dist_{r_{\delta}}(p,q)\geq \dist_{r_{\delta}}(p,S_{\delta})$ and because $d_{r_{\delta}}(q)\geq 1-\delta$ by (A1). This shows the first inequality in (\ref{DISTF}). To show the second consider a point $q'\in S_{\delta}$ such that $\dist_{r_{\delta}}(p,q')=\dist_{r_{\delta}}(p,S_{\delta})$. Then 
\ben
d_{r_{\delta}}(p)\leq \dist_{r_{\delta}}(p,q')+d_{r_{\delta}}(q')\leq \dist_{r_{\delta}}(p,S_{\delta})+1+\delta
\een
because $d_{r_{\delta}}(q')\leq 1+\delta$ by (A1) again. This shows the second inequality in (\ref{DISTF}) as wished.

Finally we prove that $r^{*}/r_{\delta}\geq 1$. To see this use (\ref{DISTF}) with $p=p^{*}$ to get $t^{*}_{\delta}+1-\delta \leq r^{*}/r_{\delta}$ and the recall that $t^{*}_{\delta}\geq 1/2$ and $\delta\leq 1/2$.
\end{proof}  

\n The Proposition \ref{COROL} shows that one can always restrict the domain of a given end and then scale the metric out to obtain an end with better asymptotic properties. More concretely one can alway cut out $E$ at $S_{\delta}$ and then define an new end consisting of the resulting unbounded region $E_{S_{\delta}}$ and the scaled metric $g_{r_{\delta}}$. Of course if this new end is AF with Schwarzschidian decay so is the original end $E$.
This shows that in order to prove asymptotic flatness for WAF ends 
it is enough to prove it for $\varepsilon$-flat ends defined as follows.

\begin{Definition}\label{DEFEF} Let $\varepsilon$ be a number in $(0,1/2)$. Then, a stationary end $E$ is $\varepsilon$-flat if it is $WAF$ and moreover,
\begin{enumerate}[labelindent=\parindent,leftmargin=*,label={\rm (H\arabic*)}]
\item The distance function $d(p)$ is smooth and every level set $S(r)=\{p,d(p)=r\}$ is an embedded sphere, and
\item For every $p\in S(r)$ we have 
\begin{align*}
& |\widehat{\Theta}(p)|\leq \frac{\varepsilon}{1+r},\qquad |\theta(p)-\frac{2}{1+r}|\leq \frac{\varepsilon}{1+r}\qquad \text{and}\qquad |Ric(p)|\leq \frac{\varepsilon}{(1+r)^{2}}
\end{align*}
\end{enumerate}
\end{Definition}

\n Note that with respect to Proposition \ref{COROL} we have changed notation from $t$ to $r$. Also note that on $\varepsilon$-flat ends, metric annuli are indeed annuli (i.e. diffeomorphic to $\AR(1,2)$, c.f. Section \ref{NOTATION}).

\section{Standard fall off for $\varepsilon$-flat ends}\label{SKTD}

As $\varepsilon$-flat stationary ends are just stationary ends with some additional properties we will continue using the same notations that we have used until now. 

\vs
\begin{Proposition}\label{PP} Let $E$ be an $\varepsilon$-flat end. Then, for all $p$ with $d(p)\geq 2/\varepsilon - 1$ we have 
\be\label{DDD}
\bigg|\big(\Delta d\big)(p)-\frac{2}{d(p)}\bigg|\leq \frac{2\varepsilon}{d(p)},
\ee
\end{Proposition} 

\begin{proof}[\bf Proof] Recall that $\Delta d=\theta$. Then, at a point $p$ with $d(p)\geq 2\varepsilon - 1$ we compute 
\begin{align*} 
\bigg|\big(\Delta d\big)(p)-\frac{2}{d(p)}\bigg|& \leq \bigg|\big(\Delta d\big)(p)-\frac{2}{d(p)+1}\bigg|+\bigg|\frac{2}{1+d(p)}-\frac{2}{d(p)}\bigg|\\
& \leq \frac{\varepsilon}{1+d(p)}+\frac{2}{(1+d(p))d(p)} \leq \big(\frac{d(p)}{1+d(p)}+\frac{2}{\varepsilon (1+d(p))}\big)\frac{\varepsilon}{d(p)}\\
& \leq \frac{2\varepsilon}{d(p)}
\end{align*}
where to obtain the second inequality we use (H2) and to obtain the last we use that $d(p)\geq 2/\varepsilon-1$.
\end{proof}

An important conclusion coming out of this Proposition is that if we let $\alpha:=1-2\varepsilon$ then the function $1/d^{\alpha}$ is superharmonic on the region $\{p/d(p)\geq 2/\varepsilon - 1\}$, namely, 
\ben
\bigg(\Delta \frac{1}{d^{\alpha}}\bigg)(p)\leq 0
\een 
for any $p$ such that $d(p)\geq 2/\varepsilon - 1$. To see this we compute
\ben
\Delta\frac{1}{d^{\alpha}}=\frac{\alpha(\alpha+1)}{d^{\alpha+2}}-\frac{\alpha}{d^{\alpha+1}}\Delta d
\een
and then use
\ben
\Delta d\geq \frac{2-2\varepsilon}{d}
\een
which is deduced from (\ref{DDD}), to obtain 
\ben
\Delta \frac{1}{d^{\alpha}}\leq \frac{\alpha(2\varepsilon+(\alpha-1))}{d^{2+\alpha}}=0
\een
as wished. We will use this below to deduce an important property of the scalar curvature $R$. For any $r\geq 0$ let 
\ben
\overline{R}(r):=\sup\big\{R(p),p\in S(r)\big\}
\een
 be the supremum of $R$ over $S(r)$. 
As the scalar curvature $R$ decays quadratically at infinity then so does $\overline{R}(r)$. For this reason if $\overline{R}(r)$ is not monotonically decreasing, namely if there are $r_{1}<r_{2}$ such that $\overline{R}(r_{1})< \overline{R}(r_{2})$, then $R$ must have a local maximum somewhere. But as $\Delta R\geq R^{2}$ such local maximum cannot exist. We conclude that $\overline{R}(r)$ must be monotonically decreasing in $r$. In particular if $\overline{R}(r_{1})=0$ for some $r_{1}$ then it is also zero for any $r\geq r_{1}$ in which case the stationary solution is simply a piece of the Minkowski space-time. We will assume therefore from now on that $\overline{R}(r)>0$ for all $r\geq 0$.

\begin{Proposition}\label{PRDEC} Let $E$ be an $\varepsilon$-flat end and let $\alpha=1-2\varepsilon$. Then, for any $r$ and $\bar{r}$ with $  2/\varepsilon - 1\leq  \bar{r}\leq r$ we have
\be\label{EFF}
\overline{R}(r)\leq \bigg(\frac{\bar{r}}{r}\bigg)^{\alpha}\overline{R}(\bar{r})
\ee
\end{Proposition}
\begin{proof}[\bf Proof] As explained before the function $1/d^{\alpha}$ is superharmonic on the region $\{p/d(p)\geq \bar{r}\}$ and therefore so is $1/d_{\bar{r}}^{\alpha}$. Moreover the function $1/d_{\bar{r}}^{\alpha}$ is identically one on $S(\bar{r})=\partial\{p,d(p)=\bar{r}\}$ and decreases to zero at infinity. 
On the other hand as $R$ is subharmonic so is $R/\bar{R}(\bar{r})$. Moreover the function $R/\bar{R}(\bar{r})$ is less or equal than one all over the set $S(\bar{r})$ and tends to zero at infinity.
We can then compare the functions $1/d_{\bar{r}}^{\alpha}$ and $R/\bar{R}(\bar{r})$ on the region $\{p/d(p)\geq \bar{r}\}$ using the maximum principle to conclude that $1/d_{\bar{r}}^{\alpha}$ is everywhere greater or equal than $R/\overline{R}(\bar{r})$ on $\{p/d(p)\geq \bar{r}\}$. Hence
\ben
R(p)\leq \bigg(\frac{\bar{r}}{r}\bigg)^{\alpha} \overline{R}(\bar{r})
\een
from which (\ref{EFF}) follows.
\end{proof}

\begin{Proposition}\label{PREL} Let $E$ be an $\varepsilon$-flat end. Then, there is a constant $c>0$ for which the following statement holds: for any $k\geq 1$ there is $r_{k}$ such that for any $\bar{r}$ and $\hat{r}$, with $r_{k}\leq \bar{r}$ and $\bar{r}\leq \hat{r}\leq 2^{k}\bar{r}$, we have
\be\label{RI}
\overline{R}(\hat{r})\leq c\bigg(\frac{\bar{r}}{\hat{r}}\bigg)^{4}\overline{R}(\bar{r})
\ee
\end{Proposition}
\n This proposition is the basis to show that the scalar curvature $R$ has a $4-\eta$ decay for any $\eta>0$. Note that the constant $c$ is independent of $k$. We prove such decay in the following Lemma. The proof of Proposition (\ref{PREL}) is given after proving the Lemma and the auxiliary Proposition \ref{PREPREL}.

\begin{Lemma}\label{4ME} {\rm [$(4-\eta)$-decay]} Let $E$ be an $\varepsilon$-flat end. Then, given $\eta>0$ there exists $c_{\eta}>0$ such that 
\be\label{FD}
\overline{R}(r)\leq \frac{c_{\eta}}{r^{4-\eta}}
\ee
for any $r>0$. In particular at any point $p$ we have
\ben
R(p)\leq \frac{c_{\eta}}{d(p)^{4-\eta}}
\een
\end{Lemma} 

\begin{proof}[\bf Proof of Lemma \ref{4ME}]
Let $k\geq 1$, let $r_{k}$ be as in Proposition \ref{PREL} and let $r$ be any number greater or equal than $2^{k}r_{k}$. Write $r$ in the form
\ben
r=2^{kl}r_{*}
\een
where $l\geq 1$ is an integer and $r_{*}\in [r_{k},2^{k}r_{k}]$. For every $m$ from $1$ to $l$ obtain
\ben
\overline{R}(2^{mk}r_{*})\leq \frac{c}{2^{4k}}\overline{R}(2^{(m-1)k}r_{*})
\een
by using Proposition \ref{PREL} with $(\bar{r},\hat{r})=(2^{(m-1)k}r_{*},2^{mk}r_{*})$ and directly from them deduce 
\be\label{A}
\overline{R}(r)=\overline{R}(2^{kl}r_{*})\leq \bigg(\frac{c}{2^{4k}}\bigg)^{l}\overline{R}(r_{*})=c^{l}\bigg(\frac{r_{*}}{r}\bigg)^{4}\overline{R}(r_{*}).
\ee
Note that $l=(\ln r/r_{*})/(k\ln 2)$ and that therefore
\be\label{CI}
c^{l}=\bigg(\frac{r}{r_{*}}\bigg)^{\displaystyle (l \ln c)/(\ln r/r_{*})}=\bigg(\frac{r}{r_{*}}\bigg)^{\displaystyle (\ln c)/(k \ln 2)}
\ee
Plugging (\ref{CI}) in (\ref{A}) we obtain 
\be\label{SUPO}
\overline{R}(r)\leq \bigg(\frac{r_{*}}{r}\bigg)^{\displaystyle 4-(\ln c)/(k\ln 2)}\overline{R}(r_{*})
\ee
Now chose $k$ big enough to have $\eta \geq (\ln c)/(k\ln 2)$. With this choice of $k$ and as $r_{*}/r\leq 1$ we obtain from (\ref{SUPO}) the inequality
\be\label{SUP}
\overline{R}(r)\leq \bigg(\frac{r_{*}}{r}\bigg)^{\displaystyle 4-\eta}\overline{R}(r_{*})
\ee
which is valid as long as $r\geq 2^{k}r_{k}$. Define now $c_{\eta}:=\max\{r_{*}^{4-\eta}\overline{R}(r_{*}),r_{*}\in [r_{k},2^{k}r_{k}]\}$. With this choice of $c_{\eta}$ the inequality (\ref{SUP}) implies (\ref{FD}) for $r\geq 2^{k}r_{k}$. 
Increase finally $c_{\eta}$ if necessary to have (\ref{FD}) valid also for $r\in (0,2^{k}r_{k})$.
\end{proof}

The following proposition will be used only inside the proof of Proposition \ref{PREL} and is given separately for the sake of a smoother exposition.   

\begin{Proposition}\label{PREPREL} Let $\psi$ be a one form in $\mathbb{R}^{3}\setminus B_{\mathbb{R}^{3}}(o,2)$, solution of
\be\label{TRI}
\left\{
\begin{array}{l}
{\rm div}\, \psi = 0,\\
{\rm d}\, \psi =0
\end{array}
\right.
\ee
and satisfying that,
\begin{enumerate}[labelindent=\parindent,leftmargin=*,label={\rm (a\arabic*')}]
\item  $|\psi (x)|\leq 2/|x|^{\alpha}$ for any $x$ in $\mathbb{R}^{3}\setminus B_{\mathbb{R}^{3}}(o,2)$, and for some $\alpha>0$, and,

\item The $C^{1}$-norm of $\psi$ over $\AR[2,4]$
% B_{\mathbb{R}^{3}}(o,4)}\setminus B_{\mathbb{R}^{3}}(o,2)$ 
is bounded by $c^{*}$,
\end{enumerate}
Then, there is a constant $\hat{c}$ depending only on $c^{*}$ such that for all $x\in \mathbb{R}^{3}\setminus B_{\mathbb{R}^{3}}(o,2)$ we have,
\be\label{STAR}
|\psi(x)|\leq \frac{\hat{c}}{4|x|^{2}} 
\ee
\end{Proposition}

The factor $1/4$ in (\ref{STAR}) and the radii $2$ and $4$ of the balls do not play any important role in the proposition but will be algebraically convenient when we use Proposition \ref{PREPREL} in Proposition \ref{PREL}.

\begin{proof}[\bf Proof.] Consider the real function $\hat{\varsigma}:\mathbb{R}\rightarrow \mathbb{R}$ defined by
\ben
\hat{\varsigma}(y)=\left\{
\begin{array}{lcl}
0 & \text{if} & y\in (-\infty,2]\cup [4,\infty),\vs\\
e^{\displaystyle -1/(y-2)(4-y)} & \text{if} & y\in (2,4),\\
\end{array}
\right.
\een
and then define the function $\varsigma:\mathbb{R}\rightarrow \mathbb{R}$ as
\ben
\varsigma(y)=\frac{\displaystyle \int_{-\infty}^{y} \hat{\varsigma}(\bar{y})\, {\rm d}\bar{y}}{\displaystyle \int_{-\infty}^{\infty} \hat{\varsigma}(\bar{y})\, {\rm d}\bar{y}}
\een
The function $\varsigma(y)$ is just a non-negative $C^{\infty}$ function taking the value zero for $y\leq 2$ and the value one for $y\geq 4$. 
Then consider a potential function $\phi$ for $\psi$ on $\mathbb{R}^{3}\setminus B_{\mathbb{R}^{3}}(o,2)$ which is simply found by integration along paths and which is unique up to a constant. Then, consider the function $\hat{\phi}(x):=\varsigma(|x|) \phi(x)$ as a function in the whole $\mathbb{R}^{3}$. This function satisfies 
\be\label{NONH}
\Delta \hat{\phi}=f
\ee
where $f$ is a function with support in $\overline{B_{\mathbb{R}^{3}}(o,4)}$ and whose $C^{1}$-norm is bounded by a constant $c_{1}^{*}(c^{*})$. One can represent then $\hat{\phi}$ as the sum of a harmonic function $\hat{\phi}_{H}$ in $\mathbb{R}^{3}$ plus the solution $\hat{\phi}_{G}$ to 
(\ref{NONH}) found by convoluting $f$ against the Green function of the Laplacian. The function $\hat{\phi}_{G}$ satisfies 
\be\label{SI}
|\hat{\phi}_{G}(x)|\leq \frac{c_{2}^{*}}{1+|x|}\qquad\text{and}\qquad |d\hat{\phi}_{G}(x)|\leq \frac{c_{3}^{*}}{(1+|x|)^{2}}
\ee
where $c_{2}^{*}$ and $c_{3}^{*}$ depend only on $c^{*}$. Thus, as $\psi(x)={\rm d}\hat{\phi}_{G}(x)+{\rm d}\hat{\phi}_{H}(x)$ when $|x|\geq 4$ and as by (a1') $|\psi(x)|\rightarrow 0$ when $|x|\rightarrow \infty$, we conclude that $|d\hat{\phi}_{H}(x)|\rightarrow 0$ when $|x|\rightarrow \infty$. In particular the harmonic functions $\partial_{x^{i}}\hat{\phi}_{H}$, $i=1,2,3$, decay also to zero at infinity. By Liouville's theorem the functions $\partial_{x^{i}}\hat{\phi}_{H}$ must be identically zero, from which we conclude that $\hat{\phi}_{H}$ is a constant and that $\psi(x)={\rm d}\hat{\phi}_{G}(x)$ when $|x|\geq 4$. Define now $\hat{c}:=\max\{4c_{3}^{*},64c^{*}\}$. Then, by (\ref{SI}), if $|x|\geq 4$ we have $|\psi(x)|\leq c_{3}^{*}/|x|^{2}\leq \hat{c}/4|x|^{2}$, while, by (a2'), if $2\leq |x|\leq 4$ we have $|\psi(x)|\leq c^{*}\leq c/4|x|^{2}$. Thus $|\psi(x)|\leq c/4|x|^{2}$ for any $|x|\geq 2$ as wished. 
\end{proof}

We have all what is necessary to prove the Proposition \ref{PREL}.

\vs
\begin{proof}[\bf Proof of Proposition \ref{PREL}] In all what follows we assume $k\geq 1$ to be given and fixed. The proof of the proposition will rely on the use of the Ernst equation. The strategy of proof will be better explained once we state and prove the facts (I), (II) and (III) below. 
\begin{enumerate}[labelindent=\parindent,leftmargin=*,label={\rm (\, \arabic*)}]
\item[(I)] {\it For every given integers $i\geq 2$ and $j\geq 0$ and divergent sequence $\bar{r}_{m}\rightarrow \infty$ there is a sequence of annuli $\Omega_{m}\supset {\mathcal A}_{\bar{r}_{m}}(1/2,2^{k+j+2})$ such that $(\Omega_{m},g_{\bar{r}_{m}})$ 
converges in $C^{i}$ to the flat annulus $(
\AR(1/2,2^{j+k+2})
,g_{\mathbb{R}^{3}})$. In particular
\be\label{CD}
\sup \big\{R_{\bar{r}_{m}}(p)/p\in {\mathcal A}_{\bar{r}_{m}}(1,2^{k+j+2})\big\}\rightarrow 0
\ee  
Moreover the complex one-forms $\zeta=2(u\nabla u+i\omega)/u^{2}$ {\rm (}restricted to ${\mathcal A}_{\bar{r}_{m}}(1/2,2^{k+j+2})${\rm )} converge in $C^{i}$ to zero.}  

\vs
\n The first part is just the definition of WAF end with $l=k+j+2$. The second part instead was discussed in the {\sc elliptic estimates} in Section \ref{NOTATION}.

\item[(II)] {\it For every given integer $j\geq 0$ there is $\tilde{r}_{j}$ such that for every $\bar{r}\geq \tilde{r}_{j}$ the following Harnak-type of estimate holds} 
\be\label{HARNAK}
1\leq \frac{\sup\big\{u(p)/p\in {\mathcal A}_{\bar{r}}(1,2^{k+j+2})\big\}}{\inf \big\{u(p)/p\in {\mathcal A}_{\bar{r}}(1,2^{k+j+2})\big\}}\leq 2
\ee

\n This is deduced from the bound $|\nabla \ln u|_{\bar{r}}^{2}\leq R_{\bar{r}}/2$ (contract the first equation in (\ref{MEE})) and from (I) by the following argument. Let $p_{1}$ and $p_{2}$ be two arbitrary points in ${\mathcal A}_{\bar{r}}(1,2^{k+j+2})$ and let $\gamma(s)$, $s\in [s_{1},s_{2}]$, be a curve inside an annulus $\Omega_{\bar{r}}\supset {\mathcal A}_{\bar{r}}(1,2^{k+j+2})$ joining $p_{1}$ to $p_{2}$ and parametrized by the $g_{\bar{r}}$-arc-length. Then one has
\begin{align}\label{37}
\bigg| \ln\frac{u(p_{2})}{u(p_{1})}\bigg| & = \bigg| \int_{s_{1}}^{s_{2}} \nabla_{\dot{\gamma}} \ln u\, {\rm d}s \bigg| \leq \int_{s_{1}}^{s_{2}} | \nabla \ln u|_{\bar{r}}\, {\rm d}s\\
\nonumber & \leq \bigg(\sup_{p\in 
\Omega_{\bar{r}}}
\sqrt{\frac{R_{\bar{r}}(p)}{2} }\bigg) {\rm length}_{\bar{r}}(\gamma)
\end{align} 
By (I) there is $\tilde{r}_{j}$ such that for any $\bar{r}\geq \tilde{r}_{j}$, there is an annulus $\Omega_{\bar{r}}\supset {\mathcal A}_{\bar{r}}(1,2^{k+j+2})$ with $\big(\Omega_{\bar{r}},g_{\bar{r}}\big)$ is sufficiently close to $(
\AR(1,2^{k+j+2}),g_{\mathbb{R}^{3}})$ in $C^{4}$ that any two points in ${\mathcal A}_{\bar{r}}(1,2^{k+j+2})$ can be joined through a curve in it of $g_{\bar{r}}$-length less or equal than $2(2\pi + 2^{k+j+2})$ (by a coarse estimation). By (\ref{CD}), if $\tilde{r}_{j}$ is big enough then 
\ben
\sup \big\{R_{\bar{r}}(p)/p\in {\mathcal A}_{\bar{r}}(1,2^{k+j+2})\big\}\leq \frac{\ln 2}{2(2\pi + 2^{k+j+2})}
\een
From this and  (\ref{37}) we obtain that for any $p_{1}$ and $p_{2}$ in ${\mathcal A}_{\bar{r}}(1/2,2^{k+j})$ we have $|\ln u(p_{2})/u(p_{1})|\leq \ln 2$. The equation (\ref{HARNAK}) then follows.

\item[(III)] {\it For any $\bar{r}\geq \tilde{r}_{j}$, with $\tilde{r}_{j}$ as in (II), and for any $p\in {\mathcal A}_{\bar{r}}(1,2^{k+j+2})$ we have
\be\label{III}
\frac{1}{2}|\chi_{\bar{r}}(p)|^{2}_{\bar{r}}\leq \frac{R(p)}{\overline{R}(\bar{r})}\leq 2 |\chi_{\bar{r}}(p)|^{2}_{\bar{r}}
\ee
where $\chi_{\bar{r}}$ is the following scaling of $\chi=2(udu+i\omega)=d{\mathcal E}$
\ben
\chi_{\bar{r}}:=\frac{\chi}{\overline{|\chi|}_{\bar{r}}},
\qquad \text{with}\qquad \overline{|\chi|}_{\bar{r}}:=\sup\big\{|\chi (p)|_{\bar{r}},p\in S(\bar{r})\big\}
\een
Moreover if $\bar{r}\geq \max\{\tilde{r}_{j},2/\varepsilon -1\}$ then 
\be\label{SOK}
|\chi_{\bar{r}}(p)|_{\bar{r}}^{2}\leq 2\bigg(\frac{1}{d_{\bar{r}}(p)}\bigg)^{\alpha}
\ee
for all $p\in {\mathcal A}_{\bar{r}}(1,2^{k+j+2})$ and where $\alpha=1=2\varepsilon$.}

\vs
\n To see the first inequality in (\ref{III}) operate as follows 
\begin{align}
\frac{R(p)}{\overline{R}(\br)}&=\frac{\displaystyle{\frac{|\chi(p)|^{2}}{2u^{4}(p)}}}{\displaystyle{\sup_{\bar{p}\in S(\br)}
\frac{|\chi(\bar{p})|^{2}}{2u^{4}(\bar{p})}} }
=\frac{\displaystyle{|\chi(p)|^{2}}}{\displaystyle{\sup_{\bar{p}\in S(\br)}|\chi(\bar{p})|^{2}\frac{u^{4}(p)}{u^{4}(\bar{p})}}}\geq\frac{1}{2}\frac{|\chi(p)|^{2}}{\displaystyle{\sup_{\bar{p}\in S(\br)}|\chi(\bar{p})|^{2}}}\\
\nonumber &\label{QQ}=\frac{1}{2}\frac{|\chi(p)|^{2}_{\br}}{\overline{|\chi|^{2}_{\br}}(\br)}=\frac{1}{2}|\chi_{\br}(p)|^{2}_{\br}
\end{align}
where in the third step we use that $u(p)/u(\bar{p})\leq 2$ by (\ref{HARNAK}). To obtain the second inequality in (\ref{III}) instead use in the third step that $u(p)/u(\bar{p})\geq 1/2$ by (\ref{HARNAK}) too. The inequality (\ref{SOK}) is the consequence of combining (\ref{EFF}) and (\ref{III}).

\end{enumerate}

\vs
With (I), (II) and (III) at hand we are in a better position to explain the strategy of proof of the proposition. The idea is to use the Ernst equation in the form (\ref{LSCH}) to show that there is a constant $\hat{c}$ independent of $k$ and a $r_{k}\geq \tilde{r}_{0}$ (here $\tilde{r}_{0}$ is $\tilde{r}_{j}$ with $j=0$) such that for every $\bar{r}\geq r_{k}$ and $p\in {\mathcal A}_{\bar{r}}[2,2^{k+1}]$ we have
\be\label{ERNST1}
|\chi_{\bar{r}}(p)|_{\bar{r}}^{2}\leq \frac{\hat{c}}{2}\bigg(\frac{\bar{r}}{d(p)}\bigg)^{4}
\ee
Together with (\ref{III}), this would imply that for every $\bar{r}\geq r_{k}$ and $p\in {\mathcal A}_{\bar{r}}[2,2^{k+1}]$ the inequality
\ben
\frac{R(p)}{\overline{R}(\bar{r})}\leq \hat{c}\bigg(\frac{\bar{r}}{d(p)}\bigg)^{4}
\een
must hold. Letting $c:=\max\{\hat{c},16\}$, then, by (\ref{EFF}), we would also have  
\ben
\frac{R(p)}{\overline{R}(\bar{r})}\leq 1\leq c\bigg(\frac{\bar{r}}{d(p)}\bigg)^{4}
\een
for every $p\in {\mathcal A}_{\bar{r}}[1,2]$. Thus, $R(p)/\overline{R}(\bar{r})\leq c(\bar{r}/d(p))^{4}$
would hold for every $p\in {\mathcal A}_{\bar{r}}[1,2^{k}]$. Hence, if $\bar{r}\leq \hat{r}\leq 2^{k}\bar{r}$ we would have $\overline{R}(\hat{r})\leq c(\bar{r}/\hat{r})^{4} \overline{R}(\bar{r})$ as wished. 

We move now to prove (\ref{ERNST1}). In the two equations of (\ref{LSCH}) divide $\chi$ by $\overline{|\chi|_{\bar{r}}}$ (which amounts to make $\chi\rightarrow \chi_{\bar{r}}$) and then in the first equation of (\ref{LSCH}) multiply both sides by $\bar{r}^{2}$ (which amounts to make $g\rightarrow g_{\bar{r}}$). In this way one obtains the equivalent system
\be\label{SERNST}
\left\{\begin{array}{l}
{\rm div}_{\bar{r}}\, \chi_{\bar{r}}=\langle\zeta,\chi_{\bar{r}}\rangle_{\bar{r}},\\
{\rm d}\, \chi_{\bar{r}}=0
\end{array}
\right.
\ee
where, recall, $\zeta=2(u\nabla u+\omega i)/u^{2}$. To deduce (\ref{ERNST1}) we will think (\ref{SERNST}) as a linear elliptic system in the variable $\chi_{\bar{r}}$ and we will consider $\zeta$ as a coefficient. 

From (I) we deduce that for any $\bar{r}\geq \breve{r}_{j}$, with $\breve{r}_{j}>0$ sufficiently big, 
there is around every $p\in {\mathcal A}_{\bar{r}}[2,2^{k+j+1}]$ an harmonic coordinate system $\{x=(x^{1},x^{2},x^{3}), |x|\leq I_{1}\}$, $p=(0,0,0)$, where the system (\ref{SERNST}) is written in the form
\ben
a^{ij}\partial_{i}\partial_{j} \phi_{\bar{r}} + b^{i}\partial_{i} \phi_{\bar{r}}=0,\qquad \partial_{i}\phi_{\bar{r}}=\chi_{\bar{r}}(\partial_{i})
\een
with the coefficients $a^{ij}$ uniformly elliptic and uniformly bounded (i.e. by $\bar{r}$-independent bounds) in $C^{4}$ and the coefficients $b^{i}$ uniformly bounded in $
C^{2}$ (of course more is known but these bounds are enough). On the other hand from (\ref{SOK}) and recalling that we find $\phi_{\bar{r}}$ from $d\phi_{\bar{r}}=\chi_{\bar{r}}$, we deduce that $\phi_{\bar{r}}$ is also uniformly bounded in $C^{1}$ if $\phi_{\bar{r}}$ is set to be zero at $(0,0,0)$.   
We can then rely in standard interior elliptic estimates over everyone of such coordinate systems to conclude that there is a constant $c^{*}>0$ such that for any $\bar{r}\geq \breve{r}_{j}$ with $\breve{r}_{j}$ sufficiently big, the $C^{2}_{g_{\bar{r}}}$-norm of $\chi_{\bar{r}}(=d\phi_{\bar{r}})$ over ${\mathcal A}_{\bar{r}}[2,2^{k+j+1}]$ is bounded by $c^{*}$. 

From these uniform $C^{2}_{g_{\bar{r}}}$-bounds for $\chi_{\bar{r}}$ and (I) the following fourth fact is just the result of a standard limit.
\begin{enumerate}[labelindent=\parindent,leftmargin=*,label={\rm (\, \arabic*)}]
\item[(IV)] {\it Given $j\geq 0$ and a divergent sequence $\{\bar{r}_{m}\}$, there is a subsequence (indexed again by $m$) such that, 
the forms $\chi_{\bar{r}_{m}}$ over ${\mathcal A}_{\bar{r}_{m}}[2,2^{k+j+1}]$ converge in $C^{1}$ to a form $\psi$ on 
$\AR(2,2^{j+k+1}]$
%$\overline{B_{\mathbb{R}^{3}}(o,2^{j+k+1})}\setminus B_{\mathbb{R}^{3}}(o,2)$ 
solution of 
\be\label{TRI}
\left\{
\begin{array}{l}
{\rm div}\, \psi = 0,\\
{\rm d}\, \psi =0
\end{array}
\right.
\ee
and satisfying that, 
\begin{enumerate}[labelindent=\parindent,leftmargin=*,label={\rm (\arabic*)}]
%\item[\rm (a')] There is a point $x$ in $\partial B_{\mathbb{R}^{3}}(o,1)$ such that $|\psi(x)|=1$,
\item[\rm (a)] $|\psi (x)|\leq 2/|x|^{\alpha}$ for every $x$ in 
$\AR[2,2^{k+j+1}]$
%$\overline{B_{\mathbb{R}^{3}}(o,2^{k+j+1})}\setminus B_{\mathbb{R}^{3}}(o,2)$
, where $\alpha=1-2\varepsilon$, and that,
\item[\rm (b)] The $C^{1}$-norm of $\psi$ over 
$\AR[2,4]$
%$\overline{B_{\mathbb{R}^{3}}(o,4)}\setminus B_{\mathbb{R}^{3}}(o,2)$
is bounded by $c^{*}$, where $c^{*}$ is the constant defined before.
\end{enumerate}  
}
\end{enumerate}

\vs

The limit form $\psi$ in (IV) satisfies then the hypothesis of Proposition \ref{PREPREL}. Therefore the Proposition \ref{PREPREL} provides us with a constant $\hat{c}(c^{*})$ that is the one that we will use now to show (\ref{ERNST1}). 
Recall that the purpose is to show that there is a constant $\hat{c}$ independent of $k$ and $r_{k}>0$ such that for every $\bar{r}\geq r_{k}$ and $p\in {\mathcal A}_{\bar{r}}[2,2^{k+1}]$ the equation (\ref{ERNST1}) holds. Having chosen $\hat{c}$ as we did, the existence of $r_{k}$ is shown by contradiction. We do that in what follows. Suppose then that there is a divergent sequence $\bar{r}_{m}\rightarrow \infty$ and a sequence of points $p_{m}\in {\mathcal A}_{\bar{r}_{m}}[2,2^{k+1}]$ such that
\ben
|\chi_{\bar{r}_{m}}(p_{m})|^{2}_{\bar{r}_{m}}> \frac{\hat{c}}{2}\bigg(\frac{\bar{r}_{m}}{d(p_{m})}\bigg)^{4}=\frac{\hat{c}}{2\, d^{\, 4}_{\bar{r}_{m}}(p_{m})}
\een
By (IV) with $j=1$, there is a subsequence (indexed again by $m$) such that the forms $\chi_{\bar{r}_{m}}$ over ${\mathcal A}_{\bar{r}_{m}}[2,2^{k+2}]$ converge in $C^{1}$ to a form $\psi$ on 
$\AR[2,2^{k+2}]$
%$\overline{B_{\mathbb{R}^{3}}(o,2^{k+2})}\setminus B_{\mathbb{R}^{3}}(o,2)$
, solution of (\ref{TRI}), satisfying (a), (b) and for which, in addition, there is a point $x_{1}\in  
\AR[2,2^{k+1}]$
%\overline{B_{\mathbb{R}^{3}}(o,2^{k+1})}\setminus B_{\mathbb{R}^{3}}(o,2)$ 
with
\ben
|\psi(x_{1})|^{2}\geq \frac{\hat{c}}{2|x_{1}|^{4}}
\een
To this subsequence one can again apply (IV) with $j=2$, to conclude that there is again a subsequence of it (indexed again by $m$) such that the forms $\chi_{\bar{r}_{m}}$ over ${\mathcal A}_{\bar{r}_{m}}[2,2^{k+3}]$ converge in $C^{2}$ to a form $\psi$ on 
$\AR[2,2^{k+3}]$
%$\overline{B_{\mathbb{R}^{3}}(o,2^{k+3})}\setminus B_{\mathbb{R}^{3}}(o,2)$, 
solution of (\ref{TRI}), satisfying (a), (b) and for which, in addition, there is a point $x_{2}\in  
\AR[2,2^{k+1}]$
%\overline{B_{\mathbb{R}^{3}}(o,2^{k+1})}\setminus B_{\mathbb{R}^{3}}(o,2)$ 
with
\ben
|\psi(x_{2})|^{2}\geq \frac{\hat{c}}{2|x_{2}|^{4}}
\een
One can continue applying (IV) with $j=3,4,\ldots$ and then taking a diagonal sequence to conclude that there is a form $\psi$ on $\mathbb{R}^{3}\setminus B_{\mathbb{R}^{3}}(o,2)$ solution of (\ref{TRI}) , satisfying (a), (b) and for which, in addition, there is a point $x_{\infty}\in  \AR[2,2^{k+1}]$
%\overline{B_{\mathbb{R}^{3}}(o,2^{k+1})}\setminus B_{\mathbb{R}^{3}}(o,2)$ 
with
\ben
|\psi(x_{\infty})|^{2}\geq \frac{\hat{c}}{2|x_{\infty}|^{4}}
\een
which is not possible because of how the constant $\hat{c}$ was defined.
\end{proof}

At this point the proof of asymptotic flatness and Schwarzschidian fall of is direct for, as observed by Kennefick and \'O Murchadha \cite{MR1314057}, once the curvature enjoys a $1/d^{4-\varepsilon}$ decay then it must forcefully enjoy a $1/d^{4}$ decay and the metric must have Schwarzschidian fall off. For completeness we give below the main elements of the construction.  

\begin{Proposition} Let $(E; g,\omega,u)$ be a WAF end. Then, there is a coordinate system $\{x=(x^{1},x^{2},x^{3})\}$ covering $E$ up to a compact set such that
\begin{gather}\label{PROPOPI}
\big|\, \delta_{ij}-g_{ij}\, \big|\leq K/|x|^{2},\qquad \big|\, \partial_{k} g_{ij}\, \big|\leq K/|x|^{3},\qquad \big|\partial_{m}\partial_{k}\, g_{ij}\, \big|\leq K/|x|^{4},\quad \text{and},
\end{gather}
\vspace{-.9cm}
\begin{gather}
\label{PROPOPII}
\big|\, \partial_{i} u\, \big|+\big|\, \omega_{i}\, \big| \leq K/|x|^{2}
\end{gather}
plus further progressive power-law decay for the norms of the multiple $\partial$-derivatives of $\tilde{g}_{ij}$, $u$ and $\omega_{i}$, where here $K$ is a positive constant and $|x|$ is the norm of $x=(x^{1},x^{2},x^{3})$ as a vector in $\mathbb{R}^{3}$.
\end{Proposition}

Note that $g_{ij}$ decays faster than $\tilde{g}_{ij}=u^{2}g_{ij}$ according to Definition \ref{AFDEF}. It is indeed the factor $u^{2}$ what causes a slower decay for $\tilde{g}_{ij}$.   

\begin{proof}[\bf Proof.] From (\ref{MEE}) we have $|(\nabla \ln u)(p)|^{2}\leq R(p)/2\leq c_{\eta}/2d(p)^{4-\eta}$. Then, scaling $u$ and $\omega$ if necessary, deduce (integrating along paths) that $\ln u$ goes to zero at infinity and furthermore that $|\ln u (p)| \leq c'_{\eta}/d(p)^{1-\eta/2}$. Similarly, if $\phi$ is a potential for $\omega$, then by (\ref{MEE}) we have $|\nabla \phi|^{2}\leq u^{4}R/2\leq \bar{c}_{\eta}/d^{4-\eta}$. Therefore by adding a constant if necessary the potential $\phi$ goes to zero at infinity and we have $|\phi(p)|\leq \bar{c}'_{\eta}/d(p)^{1-\eta/2}$. Summarizing, 
\begin{gather}\label{ETTA}
|\ln u(p)| + |\phi(p)| \leq \frac{c}{d^{1-\eta/2}(p)},\quad \text{and},\\
\label{ETTAI}|(\nabla \ln u)(p)|+|(\nabla \phi)(p)|\leq \frac{c}{d^{1-\eta/2}(p)}
\end{gather}
We will use below the following claim: {\it Let $1>\eta\geq 0$ (including the possibility $\eta=0$). If (\ref{ETTA}) holds then,
\be\label{CUREST}
|Ric(p)|\leq \frac{b_{1}}{d^{4-\eta}(p)},\qquad\text{and}\qquad |\nabla Ric (p)|\leq \frac{b_{1}}{d^{5-\eta}(p)}
\ee
and there is an harmonic coordinate system $\{x=(x^{1},x^{2},x^{3})\}$ on which we have 
\begin{gather}\label{PROPOPAI}
|\delta_{ij}-g_{ij}|\leq \frac{b_{2}}{|x|^{2-\eta}},\quad |\partial_{k} g_{ij}|\leq \frac{b_{2}}{|x|^{3-\eta}},\qquad |\partial_{l}\partial_{k} g_{ij}|\leq \frac{b_{2}}{|x|^{4-\eta}},\\
\label{PROPOPAII}
|\partial_{i} \ln u|\leq \frac{b_{2}}{|x|^{2-\eta/2}},\qquad \text{and}\qquad |\partial_{i} \phi|\leq \frac{b_{2}}{|x|^{2-\eta/2}}
\end{gather}
}

Let us postpone the proof of the claim until later and use it now, say with $\eta=1/2$. Then, from (\ref{MEE}), $\ln u$ and $\phi$ satisfy the equations
\ben
\Delta \ln u= f_{1} = O(\frac{1}{|x|^{7/2}})\qquad \text{and}\qquad \Delta \phi = f_{2} = O(\frac{1}{|x|^{7/2}})
\een
where we are thinking $f_{1}$ and $f_{2}$ as sources. With this fast decay of the sources we get for both, $\ln u$ and $\phi$, a decay $O(1/|x|)$, namely as in (\ref{ETTA}) with $\eta=0$ [\footnote{The reader can check that and $||x|-d(x)|\leq$ constant.}]. Then, (\ref{PROPOPAI}) and (\ref{PROPOPAII}) with $\eta=0$ are just (\ref{PROPOPI}) and (\ref{PROPOPII}) respectively. The further progressive power-law decay is achieved by a standard elliptic bootstrap of decay and is unnecessary to include here.   

To finish the proof we need to explain how to prove the claim. From (\ref{ETTA}) and for any $p\in {\mathcal A}_{r}(1/2,2)$ we have $|\phi(p)|\leq c/d(p)^{1-\eta/2} = c(r/d(p))^{1-\eta/2}(1/r)^{1-\eta/2}\leq 2c/r^{1-\eta/2}$ and $|(\nabla \phi)(p)|_{r}\leq c/d(p)^{2-\eta/2}(1/r)^{2-\eta/2} = c(r/d(p))^{2-\eta/2}(1/r)^{2-\eta/2}\leq 4c/r^{2-\eta/2}$. Similar bounds are obtained for $\ln u$. Summarizing, all over ${\mathcal A}_{r}(1/2,2)$ we have the uniform bounds 
\be\label{UNIBOUND}
|\ln u|+|\phi|\leq 4c/r^{1-\eta/2}\qquad \text{and}\qquad |\nabla \ln u|_{r} + |\nabla \phi|_{r}\leq 8c/r^{1-\eta/2}
\ee
We will use them to obtain interior elliptic estimates on the smaller annuli ${\mathcal A}_{r}(4/5,5/4)$. Recall that when $r\rightarrow \infty$ the annulus $({\mathcal A}_{r}(1/2,2),g_{r})$ converge to $(\AR(1/2,2),g_{\mathbb{R}^{3}})$ in $C^{i}$ for every $i\geq 2$ and due to this we do not need to worry about the constants involved in Sobolev embeddings or elliptic estimates if $r$ is sufficiently large. Scaling (\ref{MEE}) gives on ${\mathcal A}_{r}(1/2,2)$
\ben
\Delta_{r} \ln u = -2{\displaystyle \frac{|\nabla \phi|_{r}^{2}}{u^{4}}}=f_{1,r},\qquad \text{and}\qquad \Delta_{r} \phi = 4\langle \nabla \ln u,\nabla \phi\rangle_{r}=f_{2,r}
\een
where we will think $f_{r,1}$ and $f_{r,2}$ as sources. Then, $L^{p}$ interior elliptic estimates \cite{MR737190} give 
\ben
\|\ln u\|_{H^{2,4}_{g_{r}}({\mathcal A}_{r}(3/5,5/3))}\leq c_{1}\big(\|f_{1,r}\|_{L^{4}_{g_{r}}({\mathcal A}_{r}(1/2,2))}+\|\ln u\|_{L^{4}_{g_{r}}({\mathcal A}_{r}(1/2,2))}\big)\leq \frac{c_{2}}{r^{1-\eta/2}} 
\een
where to obtain the last inequality use (\ref{UNIBOUND}) (assume $r>1$). Sobolev embeddings then give $\|\ln u\|_{C^{1,\beta}_{g_{r}}({\mathcal A}_{r}(3/5,5/3))}\leq c_{3}/r^{1-\eta/2}$. In the same way we obtain $\|\phi\|_{C^{1,\beta}_{g_{r}}({\mathcal A}_{r}(3/5,5/3))}\leq c_{4}/r^{1-\eta/2}$. Use these bounds to get $C^{1,\beta}$ bounds for the sources $f_{1,r}$ and $f_{2,r}$ and from them and Schauder estimates \cite{MR737190} get the bound $\|\ln u\|_{C^{2,\beta}_{g_{r}}({\mathcal A}_{r}(4/5,5/4))}\leq c_{5}/r^{1-\eta/2}$ and  $\|\phi\|_{C^{2,\beta}_{g_{r}}({\mathcal A}_{r}(4/5,5/4))}\leq c_{5}/r^{1-\eta/2}$. In particular if $r=d(p)$ and by undoing the scaling we have $|(\nabla\nabla \ln u)(p)|\leq c_{5}/r^{3-\eta/2}= c_{5}/d(p)^{3-\eta/2}$ and $|(\nabla\nabla \phi)(p)|\leq c_{5}/r^{3-\eta/2}= c_{5}/d(p)^{3-\eta/2}$. Use these estimates and (\ref{ETTAI}) in (\ref{MEE}) to arrive easily at (\ref{CUREST}). Finally, as shown in \cite{MR1001844} ({\bf Theorem} in pg. 314, with $\eta$ there equal to $2-\eta$ here; see also {\it Remark} 1) cubic volume growth and the curvature decays (\ref{CUREST}) are enough to guarantee the existence of a coordinate system $\{x=(x^{1},x^{2},x^{3})\}$ satisfying (\ref{PROPOPAI}). The proof of the claim is finished. \end{proof} 

\bibliographystyle{plain}
\bibliography{Master.bib}

\begin{thebibliography}{10}

\bibitem{MR1806984}
Michael~T. Anderson.
\newblock On stationary vacuum solutions to the {E}instein equations.
\newblock {\em Ann. Henri Poincar\'e}, 1(5):977--994, 2000.

\bibitem{MR1001844}
Shigetoshi Bando, Atsushi Kasue, and Hiraku Nakajima.
\newblock On a construction of coordinates at infinity on manifolds with fast
  curvature decay and maximal volume growth.
\newblock {\em Invent. Math.}, 97(2):313--349, 1989.

\bibitem{MR737190}
David Gilbarg and Neil~S. Trudinger.
\newblock {\em Elliptic partial differential equations of second order}.
\newblock Springer-Verlag, Berlin, second edition, 1983.

\bibitem{MR1314057}
Daniel Kennefick and Niall {\'O}~Murchadha.
\newblock Weakly decaying asymptotically flat static and stationary solutions
  to the {E}instein equations.
\newblock {\em Classical Quantum Gravity}, 12(1):149--158, 1995.

\bibitem{MR2441850}
Reinhard Meinel, Marcus Ansorg, Andreas Kleinw{\"a}chter, Gernot Neugebauer,
  and David Petroff.
\newblock {\em Relativistic figures of equilibrium}.
\newblock Cambridge University Press, Cambridge, 2008.

\bibitem{Ehlers1980}
Ehlers J~(1980) %n.
\newblock Isolated systems in general relativity.
\newblock {\em Annals of the New York Academy of Sciences}, 336(1):279--294,
  1980.

\bibitem{Persic:1995ru}
Massimo Persic, Paolo Salucci, and Fulvio Stel.
\newblock {The Universal rotation curve of spiral galaxies: 1. The Dark matter
  connection}.
\newblock {\em Mon.Not.Roy.Astron.Soc.}, 281:27, 1996.

\bibitem{MR2243772}
Peter Petersen.
\newblock {\em Riemannian geometry}, volume 171 of {\em Graduate Texts in
  Mathematics}.
\newblock Springer, New York, second edition, 2006.

\bibitem{Salucci:2007tm}
Paolo Salucci, A.~Lapi, C.~Tonini, G.~Gentile, I.~Yegorova, et~al.
\newblock {The Universal Rotation Curve of Spiral Galaxies. 2. The Dark Matter
  Distribution out to the Virial Radius}.
\newblock {\em Mon.Not.Roy.Astron.Soc.}, 378:41--47, 2007.

\bibitem{Sofue:2000jx}
Yoshiaki Sofue and Vera Rubin.
\newblock {Rotation curves of spiral galaxies}.
\newblock {\em Ann.Rev.Astron.Astrophys.}, 39:137--174, 2001.

\bibitem{MR2003646}
Hans Stephani, Dietrich Kramer, Malcolm MacCallum, Cornelius Hoenselaers, and
  Eduard Herlt.
\newblock {\em Exact solutions of {E}instein's field equations}.
\newblock Cambridge Monographs on Mathematical Physics. Cambridge University
  Press, Cambridge, second edition, 2003.

\bibitem{MR7571801}
J.~Binney;~S. Tremaine.
\newblock {\em Galactic Dynamics (second edition).}
\newblock Princeton University Press - New Jersey.

\bibitem{MR757180}
Robert~M. Wald.
\newblock {\em General Relativity}.
\newblock University of Chicago Press, Chicago, IL, 1984.

\end{thebibliography}

\end{document}